\documentclass[11pt]{article}
\usepackage{algorithm}
\usepackage{algpseudocode}
\usepackage{amsfonts}
\usepackage{amsmath}
\usepackage{amsthm}
\usepackage{array}
\usepackage[margin=1in]{geometry}
\usepackage{graphicx}
\usepackage{hyperref}
\hypersetup{
    colorlinks=true,
    linkcolor=blue,
    urlcolor=blue,
    citecolor = blue
}
\usepackage[utf8]{inputenc}
\usepackage[numbers]{natbib}
\usepackage[dvipsnames]{xcolor}
\newcolumntype{P}[1]{>{\centering\arraybackslash}p{#1}}

\title{Pan-Private Uniformity Testing}
\author{Kareem Amin \thanks{Google New York,  \href{mailto:kamin@google.com}{\texttt{kamin@google.com}}.} \and Matthew Joseph \thanks{Google New York, \href{mailto:mtjoseph@google.com}{\texttt{mtjoseph@google.com}}. Part of this work done while a graduate student at the University of Pennsylvania.} \and Jieming Mao \thanks{Google New York, \href{mailto:maojm@google.com}{\texttt{maojm@google.com}.}}}

\newtheorem{theorem}{Theorem}
\newtheorem{fact}{Fact}
\newtheorem{definition}{Definition}
\newtheorem{lemma}{Lemma}

\newcommand{\A}{\mathcal{A}}
\newcommand{\ber}[1]{\mathsf{Ber}\left(#1\right)}

\newcommand{\E}[2]{\mathbb{E}_{#1}\left[ #2 \right]}
\newcommand{\eps}{\varepsilon}

\newcommand{\In}{\mathcal{I}}
\newcommand{\kl}[2]{D_{KL}\left(#1||#2\right)}
\newcommand{\lap}[1]{\mathsf{Lap}\left(#1\right)}
\newcommand{\Ou}{\mathcal{O}}
\renewcommand{\P}[2]{\mathbb{P}_{#1}\left[ #2 \right]}
\newcommand{\poi}[1]{\mathsf{Poisson}\left(#1\right)}

\newcommand{\ppu}{\textsc{SimplePanTest}}
\newcommand{\ppub}{\textsc{PanTest}}

\renewcommand{\S}{\mathcal{S}}

\newcommand{\St}{\mathcal{S}}
\newcommand{\tv}[2]{||#1 - #2||_{TV}}
\newcommand{\Var}[1]{\text{Var}\left[ #1 \right]}
\newcommand{\X}{\mathcal{X}}

\newif\ifcomment
\commenttrue

\ifcomment
\newcommand{\ka}[1]{\textcolor{violet}{[KA: #1]}}
\newcommand{\mj}[1]{\textcolor{magenta}{[MJ: #1]}}
\newcommand{\jm}[1]{\textcolor{red}{[JM: #1]}}
\else
\newcommand{\ka}[1]{}
\newcommand{\mj}[1]{}
\newcommand{\jm}[1]{}
\fi

\begin{document}

\maketitle

\abstract{A centrally differentially private algorithm maps raw data to differentially private outputs. In contrast, a locally differentially private algorithm may only access data through public interaction with data holders, and this interaction must be a differentially private function of the data. We study the intermediate model of \emph{pan-privacy}. Unlike a locally private algorithm, a pan-private algorithm receives data in the clear. Unlike a centrally private algorithm, the algorithm receives data one element at a time and must maintain a differentially private internal state while processing this stream.

First, we show that pure pan-privacy against multiple intrusions on the internal state is equivalent to sequentially interactive local privacy. Next, we contextualize pan-privacy against a single intrusion by analyzing the sample complexity of uniformity testing over domain $[k]$. Focusing on the dependence on $k$, centrally private uniformity testing has sample complexity $\Theta(\sqrt{k})$, while noninteractive locally private uniformity testing has sample complexity $\Theta(k)$. We show that the sample complexity of pure pan-private uniformity testing is $\Theta(k^{2/3})$. By a new $\Omega(k)$ lower bound for the sequentially interactive setting, we also separate pan-private from sequentially interactive locally private and multi-intrusion pan-private uniformity testing.}

\thispagestyle{empty} \setcounter{page}{0}
\clearpage

\section{Introduction}
\label{sec:intro}
Differential privacy~\cite{DMNS06} promises that a randomized algorithm's output distribution is relatively insensitive to small changes in its input data. This insensitivity hides the presence or absence of individual data elements and provides privacy for the contributors of that data. Rigorous privacy guarantees have driven increasing adoption of differential privacy by industry~\cite{A17, BEMMR+17, DKY17, Go19}, government~\cite{A18}), and academic researchers~\cite{MTKV18, MDHKM+19}.

In \emph{central differential privacy}~\cite{DMNS06}, the algorithm receives a database in the clear, and privacy only constrains the algorithm's eventual output. Central privacy therefore offers the highest utility -- for example, the lowest error or sample complexity -- but weakest privacy guarantee. In particular, in many real-world applications the input database is acquired over time, and raw data is kept until the time arrives to produce (differentially private) outputs. A user may worry that raw data sitting with a trusted algorithm operator may still be at risk of exfiltration by subpoena, ``mission creep'' by the operator that contravenes users' original wishes, or a change in operator ownership. Since central privacy makes no guarantees about the intermediate representation of the data during processing, it offers no protection against these events.

One solution to this family of problems is \emph{local differential privacy}~\cite{DMNS06, KLNRS11}. Locally differentially private algorithms do not receive a database in the clear. Instead, data remains distributed among users, and the algorithm must learn about the data by interacting with these users in a public yet privacy-preserving way. Because users are in charge of randomizing their communications in the protocol, they no longer need to trust an algorithm operator. Unfortunately, this strong privacy guarantee often incurs a significant utility cost. For example, one can compute the sum of $n$ bits to $O(1/\eps)$ additive error under $\eps$-central privacy but, for constant $\eps$, must incur $\Omega(\sqrt{n})$ error under $\eps$-local privacy~\cite{CSS12}.

We study \emph{pan-privacy}~\cite{DNPRY10} as a middle ground in this tradeoff between privacy and utility. A pan-private algorithm receives a stream of raw data (for example, the gradual data acquisition process mentioned above). Pan-privacy has two requirements. First, while processing the data a pan-private algorithm must maintain an internal state that is differentially private against any single intrusion. Second, a pan-private algorithm must ultimately produce a differentially private output.

Central, pan-, and local privacy therefore correspond to different trust models. If a user trusts the algorithm operator to not only perform the computation in question but to responsibly steward raw data in the future, then central privacy is a sufficient guarantee. If a user currently trusts the operator, but also wants to protect themselves against unknown future complications in data stewardship, pan-privacy suffices. For a user who does not trust the operator at all, only local privacy is enough.

\subsection{Contributions}
\label{subsec:contribs}
We give several results about the relative merits of these models. Taken together, they suggest pan-privacy as a middle ground for both privacy and utility between the central and local models.
\begin{enumerate}
	\item Through constructive transformations in both directions, we show that pure pan-privacy against multiple intrusions is equivalent to sequentially interactive local privacy (Section~\ref{sec:pan_local}).
	\item We give matching (in $k$) upper and lower bounds showing that uniformity testing --- the problem of distinguishing uniform and non-uniform distributions through sample access --- has pure pan-private sample complexity $\Theta(k^{2/3})$. The best known locally private uniformity tester achieves $\Theta(k)$ sample complexity by reducing uniformity testing to binary testing~\cite{ACFT19}, while the optimal centrally private uniformity tester gets $\Theta(\sqrt{k})$ without reducing the problem domain at all~\cite{ASZ18}. Our pan-private uniformity tester intermediates between these approaches by reducing uniformity testing over $[k]$ to, roughly, uniformity testing over $[k^{2/3}]$ (Section~\ref{sec:uni}). Our lower bound adapts the approach used by~\citet{DGKR19} to prove testing lower bounds under memory and communication restrictions  (Section~\ref{sec:pp_lb}). 
	\item By a new lower bound, again adapting the memory-restricted lower bound of~\citet{DGKR19}, we show that sequentially interactive locally private uniformity testing has sample complexity $\Theta(k)$ (Section~\ref{sec:lp_lb}).
\end{enumerate}

We briefly elaborate on the first contribution. We view this result as dictating the scope of when (pure) pan-privacy is reasonable. If a user requires privacy against multiple intrusions, then the operator suffers no utility loss by using an algorithm that is locally private instead of an algorithm that is pan-private against multiple intrusions. However, there are cases where a user may be satisfied with pan-privacy against a single intrusion. To see why, we use the following simple result.

\begin{fact}
	Suppose a user's data is element $s_t$ of an $(\eps,\delta)$-pan-private algorithm $\A$'s stream. We say an intrusion occurs at time $t$ if the intrusion occurs immediately after $\A$ updates its internal state to $i_t$ after seeing element $s_t$. If
	\begin{enumerate}
		\item the first intrusion (possibly of many) occurs at time $t' \geq t$, or
		\item all intrusions occur at times $t' < t$,
	\end{enumerate}
then the intruder's view is an $(\eps,\delta)$-differentially private function of $s_t$. 
\end{fact}
\begin{proof}
	Pan-privacy guarantees that $i_t$ is an $(\eps,\delta)$-differentially private function of $s_t$. In Case 1, the adversary only sees a post-processing of $i_t$. Differential privacy's resilience to post-processing (see e.g. Proposition 2.1 in Dwork and Roth's survey~\cite{DR14}) implies that this view is $(\eps,\delta)$-differentially private in $s_t$. In Case 2, the adversary's view is independent of $s_t$, so $(\eps,\delta)$-differential privacy is immediate.
\end{proof}

By Fact 1, if $\A$ is pan-private against a single intrusion, then it guarantees privacy for users who either contribute data before the first intrusion or after all intrusions. However, pan-privacy is not sufficient to protect a user's privacy if the operator has already been compromised and may be compromised again. The key parameter for pan-privacy is therefore the user's trust in the operator when the user contributes their data. This motivates the trust model described in the introduction: if a user trusts the operator today, but wants to ``future-proof'' themselves for tomorrow, then pan-privacy is a reasonable privacy guarantee.

\subsection{Related Work}
\label{subsec:rel}
We start with previous work on pan-privacy.~\citet{DNPRY10} introduced pan-privacy and gave pan-private algorithms for several different counting problems over streams. They also gave two lower bounds. First they separated pan-privacy against one and two intrusions by showing that estimating the number of distinct elements in a stream is much harder with multiple intrusions. Second, they gave a problem, inner product counting, that separates pan-privacy from noninteractive local privacy.~\citet{MMNW11} extended these results to new counting problems and dynamic streams. They also showed that pan-private algorithms cannot approximate distinct element count to additive accuracy $o(\sqrt{|X|})$ for data universe $X$. This improved upon the $\Omega(\sqrt{|X|}/\log(|X|))$ lower bound given by~\citet{MMPRTV10} for two-party differential privacy (a weaker guarantee than pan-privacy), which was the first separation between central and pan-privacy.~\citet{DNPR10} also studied pan-privacy, albeit under the additional constraint of continual observation, which requires the algorithm to provide accurate answers after every stream element. They and~\citet{CSS11} gave both upper and lower bounds for counting problems under continual observation.

Our work departs from the above in a few ways. First, we generalize previous results on pan-privacy against two intrusions by showing that it is equivalent to a different model, sequentially interactive local privacy. Second, the testing problems we study focus on learning from samples generated by some distribution, as opposed to previous work on adversarial streaming problems. This distributional quality necessitates different lower bound techniques.

In uniformity testing, a line of work~\cite{GR00,P08,VV14} has established that uniformity testing (without privacy) has sample complexity $\Theta\left(\tfrac{\sqrt{k}}{\alpha^2}\right)$ where $k$ is the domain size and $\alpha$ is the total variation distance parameter (for more information on testing, see the survey by~\citet{C15}).~\citet{ASZ18} showed that $\eps$-centrally private uniformity testing has sample complexity $\Theta\left(\tfrac{\sqrt{k}}{\alpha^2} + \tfrac{\sqrt{k}}{\alpha\sqrt{\eps}} + \tfrac{k^{1/3}}{\alpha^{4/3}\eps^{2/3}} + \tfrac{1}{\alpha\eps}\right)$.~\citet{ACFT19} showed that noninteractive $\eps$-locally private uniformity testing has sample complexity $\Theta\left(\tfrac{k}{\alpha^2\eps^2}\right)$. \citet{ACHST19} proved similar results with special attention to the amount of public randomness. A comparison of our results to this previous work appears in Figure~\ref{fig:comp}. 

In the data structures community, several works have studied \emph{history independence}~\cite{M97, NT01, BG07}. A history independent data structure is one whose memory representation reveals no more information than its abstract representation does. For example, without history independence, the abstract representation of a dictionary may only reveal keys and values while the memory representation also reveals insertion order. Pan-privacy instead aims to guarantee that the abstract representation is a differentially private function of the input data.

\begin{figure*}[h!]
    \renewcommand{\arraystretch}{2}
    \centering
	\begin{tabular}{|c|P{60mm}|P{50mm}|}
    \hline
    Setting & Previous Work & This Work \\ \hline
    Non-private & $\Theta\left(\frac{\sqrt{k}}{\alpha^2}\right)$ \cite{GR00,P08,VV14} & -- \\ \hline
    $\eps$-central privacy & $\Theta\left(\frac{\sqrt{k}}{\alpha^2} + \tfrac{\sqrt{k}}{\alpha\sqrt{\eps}} + \frac{k^{1/3}}{\alpha^{4/3}\eps^{2/3}} + \frac{1}{\alpha\eps}\right)$  \cite{ASZ18}  & -- \\ \hline
    $\eps$-pan-privacy & \hspace{37.75pt} -- \newline --& $O\left(\frac{k^{2/3}}{\alpha^{4/3}\eps^{2/3}} + \frac{\sqrt{k}}{\alpha^2} + \frac{\sqrt{k}}{\alpha \eps}\right)$ $\Omega\left(\frac{k^{2/3}}{\alpha^{4/3}\eps^{2/3}}  + \frac{\sqrt{k}}{\alpha^2}  + \frac{1}{\alpha \eps}\right)$ \\ \hline
    SI $\eps$-local privacy & $O\left(\frac{k}{\alpha^2\eps^2}\right)$ \cite{ACFT19} & $\Omega\left(\frac{k}{\alpha^2\eps^2}\right)$ \\ \hline
    NI $\eps$-local privacy & $\Theta\left(\frac{k}{\alpha^2\eps^2}\right)$ \cite{ACFT19} & -- \\ \hline
  \end{tabular}
  \renewcommand{\arraystretch}{1}
	\caption{A comparison of the uniformity testing sample complexity bounds given in this and previous work. ``SI'' is sequentially interactive and ``NI'' is noninteractive. Before this work, no pan-private bounds were known, and it was not known that $O\left(\tfrac{k}{\alpha^2\eps^2}\right)$ is tight for sequentially interactive protocols.
\label{fig:comp}}
\end{figure*}
\section{Preliminaries}
\label{sec:prelims}

\subsection{Central Differential Privacy}
\label{subsec:dp}
A randomized algorithm $\A$ satisfies central differential privacy if it maps raw databases to outcomes such that the distribution over outcomes is relatively insensitive to small changes in the database. This insensitivity, which hides the presence or absence of any one user, provides the privacy guarantee.

\begin{definition}[Central differential privacy~\cite{DMNS06}]
\label{def:dp}
    Given data universe $\X$ and two databases $D, D' \in \X^n$, $D$ and $D'$ are \emph{neighbors} if they differ in $\leq 1$ element. Given algorithm $\A \colon \X^n \to Y$, $\A$ is \emph{$(\eps,\delta)$-differentially private} if for all subsets $S \subset Y$, $$\P{\A}{\A(D) \in S} \leq e^\eps \P{\A}{\A(D') \in S} + \delta.$$
\end{definition}

For this work, it is important to note that centrally private algorithms enjoy trusted (central) access to the entire raw database. In particular, they may perform arbitrary computations on raw data before releasing a private output. Pan- and locally private algorithms have restricted forms of access to the data.

\subsection{Pan-privacy}
\label{subsec:pp}
A \emph{pan-private} algorithm operates in a different setting with different guarantees. Here, the algorithm $\A$ receives the database as a stream, one element at a time, and updates its internal state after seeing each element. The element is then deleted from $\A$'s memory, and $\A$ continues processing the stream\footnote{As is standard in pan-privacy, we assume that the process of receiving an element and updating the state is atomic: the adversary cannot intrude on the internal state between the reception of a new stream element and the internal state update. Without this assumption, nothing prevents the adversary from possibly seeing a data point in the clear, and differential privacy is impossible.}. The entirety of $\A$'s knowledge about the stream so far is thus contained in this internal state. At the end of the stream, $\A$ produces an output as its final answer. Pan-privacy mandates that $\A$'s internal state and final answer must be differentially private functions of the stream on a per-element basis.

\begin{definition}[Pan-privacy~\cite{DNPRY10}]
\label{def:pan}
    Let $\X$ be a data universe, and let $\St = \X^{\mathbb{N}}$ be the set of streams from $\X$. Two streams $s, s' \in \St$ are \emph{neighbors} if there exist $x$ and $x' \in \X$ such that replacing a single instance of $x \in s$ with $x'$ produces $s'$.
    
    A pan-private algorithm consists of an internal algorithm $\A_{\In}$ and an output algorithm $\A_{\Ou}$. $\A$ maps streams to internal states by repeated application of $\A_{\In}$, which maps an internal state and element of $\X$ to an internal state, $\A_{\In} \colon \In \times \X \to \In$. At some time the stream ends and $\A$ publishes a final output $\A_{\Ou}(i)$ where $i$ is the internal state of $\A$ at the end of the stream. For stream $s$, let $\A_{\In}(s)$ denote the internal state of $\A$ after processing $s$, and let $s_{\leq t}$ denote the first $t$ elements of stream $s$. $\A$ is \emph{$(\eps,\delta)$-pan-private} if, for any neighboring streams $s$ and $s'$, any time $t$, and any set $E \subset \In \times \Ou$
    \begin{equation}
    \label{eq:pan}
    \P{\A}{(\A_{\In}(s_{\leq t}), \A_{\Ou}(\A_{\In}(s))) \in E} \leq e^\eps\P{\A}{(\A_{\In}(s_{\leq t}'), \A_{\Ou}(\A_{\In}(s'))) \in E} + \delta.
    \end{equation}
    This paper will focus on pure pan-privacy, where $\delta = 0$. We shorthand this as $\eps$-pan-privacy.
\end{definition}

Pan-privacy thus protects against an adversary that sees any single internal state of $\A$ as well as its final output. The second requirement implies that any pan-private algorithm is also centrally private; the key additional contribution of pan-privacy is the maintenance of the differentially private internal state. To generalize Definition~\ref{def:pan} to $c > 1$ intrusions, we can replace inequality~\ref{eq:pan} with 
$$\P{\A}{(\A_{\In}(s_t)_{t=t_1}^{t_c}, \A_{\Ou}(\A_{\In}(s))) \in E} \leq e^\eps\P{\A}{(\A_{\In}(s_t')_{t=t_1}^{t_c}, \A_{\Ou}(\A_{\In}(s'))) \in E} + \delta$$
where $E \subset \In^c \times \Ou$.

We note that our definition of pan-privacy differs from the original. This stems from slightly different goals. The original work of~\citet{DNPRY10} focused on tracking statistics of a stream of unknown length and allowed for the possibility that the stream could end unexpectedly. They also allowed for multiple outputs by the algorithm. We instead analyze problems from a sample complexity perspective and focus on the number of samples needed to solve a problem pan-privately. This leads us to consider streams of fixed length (determined by the algorithm as the required sample complexity) and a single output (the answer to the problem in question). Additionally,~\citet{DNPRY10} studied streams where each element is a (user, value) pair and a neighboring stream may replace all values contributed by any one user. We remove the notion of a user and simply view a stream as a sequence of elements. We therefore guarantee element-level rather than user-level privacy\footnote{Note that this allows the closest comparison with existing centrally and locally private uniformity testers, which all employ element-level privacy.}. Nonetheless, the basic idea of pan-privacy -- its privacy against an adversary who sees a single internal state and the output -- remains intact.

\subsection{Local Differential Privacy}
\label{subsec:ldp}
A \emph{locally differentially private} algorithm satisfies a still more restrictive privacy guarantee. Unlike pan-private algorithms, locally private algorithms never see any data in the clear. Instead, a locally private algorithm is a public interaction between users, each of whom privately holds a single data element. Since our main point of comparison is pan-private algorithms, we view these users as stream elements. A pan-private algorithm sees each stream element, but a locally private algorithm only sees the \emph{randomizer} output produced by each stream element. However, this is a difference only in presentation, and a user obtains the same kind of local privacy guarantee whether we view them as a user or a stream element. Up to this difference, our local differential privacy definitions generally imitate those given by~\citet{JMNR19}.

\begin{definition}
\label{def:randomizer}
	An \emph{$(\eps,\delta)$-randomizer} $R \colon X \to Y$ is an $(\eps,\delta)$-differentially private function taking a single data point as input.
\end{definition}

Because communication occurs only through randomizers, the overall record of public interaction is private. We more formally study this interaction in terms of its \emph{transcript}.

\begin{definition}
\label{def:transcript}
	A \emph{transcript} $\pi$ is a vector of tuples $(R_t, y_t)$ indicating the randomizer used and output produced at each time $t$. 
\end{definition}

We can then view a locally private protocol as a coordinating mechanism that takes a transcript and selects a randomizer for the next stream element. 

\begin{definition}
\label{def:protocol}
	Let $S_\pi$ denote the collection of transcripts and $S_R$ the collection of randomizers. Then a \emph{protocol} $\A$ is a function $\A \colon S_\pi \to S_R$ mapping transcripts to randomizers.
\end{definition}

A locally private protocol generally includes some post-processing of the transcript to generate some final output. Since this post-processing is still a function of the transcript, we abstract it away and focus only on the transcript. Next, we distinguish between different notions of interactivity for locally private protocols.

\begin{definition}[\cite{DJW13}]
\label{def:inter}
	If locally private protocol $\A$ makes all randomizer assignments before the stream begins (i.e., each randomizer choice $R_t$ is independent of the transcript so far conditioned on $t$) then $\A$ is \emph{noninteractive}. If $\A$ makes these assignments adaptively as the stream progresses, then $\A$ is \emph{sequentially interactive}.
\end{definition}

The most general model of local privacy allows \emph{full interactivity}: users may produce arbitrarily many outputs in arbitrary sequences. In particular, a protocol may re-query past participants. This is analogous to processing a stream with multiple passes. Since we focus on pan-privacy in the single-pass model, we will compare it to noninteractive and sequentially interactive locally private protocols, which can only query each participant at most once. We now formally define local differential privacy.

\begin{definition}
\label{def:ldp}
	A protocol $\A$ is \emph{$(\eps,\delta)$-locally differentially private} if its transcript is an $(\eps,\delta)$-differentially private function of the user data. If $\delta = 0$, we say $\A$ is \emph{$\eps$-locally differentially private}.
\end{definition}

In particular, a sequentially interactive protocol is $(\eps,\delta)$-locally differentially private if and only if each randomizer used is an $(\eps,\delta)$-randomizer.
\section{Pan-privacy and Local Privacy}
\label{sec:pan_local}
We first show that any algorithm that is pure pan-private against multiple intrusions has a \emph{locally} private equivalent (Theorem~\ref{thm:pan_local}). The main idea is that the operator of a pan-private algorithm $\A_{2P}$ cannot know when two intrusions will occur. In particular, if the two intrusions occur at times $t$ and $t+1$ --- respectively, immediately after $\A_{2P}$ processes $s_t$ and $s_{t+1}$ --- then failure to randomize the internal state between $t$ and $t+1$ may reveal element $s_{t+1}$. The operator must therefore re-randomize the state at \emph{every} time step.

We briefly sketch the proof of Theorem~\ref{thm:pan_local} (full proofs of this and other results appear in the Appendix). First, we observe that any $\A_{2P}$ that is $\eps$-pan-private against two intrusions can be modified into an algorithm $\A_{1P}$ that maintains all of its internal states thus far and still remain $\eps$-pan-private against \emph{one} intrusion (Lemma~\ref{lem:two_one_pp}). Because this single intrusion may come at the end of the stream, the complete list of internal states during the stream must be an $\eps$-differentially private function of the stream. We can therefore simulate this procedure in the sequentially interactive local model and have the transcript generate this complete list of internal states (Lemma~\ref{lem:one_pp_ldp}).

In the other direction, we convert any $\eps$-sequentially interactive locally private protocol $\A_L$ to $\A_{2P}$, which is $\eps$-pan-private against two intrusions. $\A_{2P}$ simulates $\A_L$ and stores the transcript so far as its internal state. Since this transcript is an $\eps$-differentially private function of the data (recall that the transcript for $\A_L$ is public), $\A_{2P}$ is $\eps$-pan-private against an arbitrary number of intrusions onto its internal state. 

\begin{theorem}
\label{thm:pan_local}
	For every $\A_{2P}$ that is $\eps$-pan-private against two intrusions and generates output distribution $O$ given input stream $s$, there exists $\A_L$ that is sequentially interactive $\eps$-locally private and generates transcript distribution $O$ given $s$, and vice-versa.
\end{theorem}
\begin{proof}
    \underline{$\Rightarrow$ (pan to local)}: We start by converting from pan-privacy against two intrusions to pan-privacy against one intrusion while preserving all internal states.
 
    \begin{lemma}
    \label{lem:two_one_pp}
        Suppose $\A_{2P}$ is $\eps$-pan-private against two intrusions, and let $I_{2,t}$ be the random variable for the internal state of $\A_{2P}$ after stream element $t$. Then there exists $\A_{1P}$ that is $\eps$-pan-private against one intrusion such that, for analogously defined $I_{1,t}$, for any stream $s_{\leq t}$, the concatenation $I_{2,1} \circ I_{2,2} \cdots \circ I_{2,t}$ is distributed identically to $I_{1,t}$. 
    \end{lemma}
    \begin{proof}
    We first define $\A_{1P}$. For $j \in \{1,2\}$, define $i_{j,t}$ to be the realized internal state of $\A_{jP}$ after seeing the $t^{th}$ stream element. Each internal state $i_{1,t}$ of $\A_{1P}$ is a concatenation of internal states $i_{2,1} \circ \cdots \circ i_{2,t}$, and for any internal state $i$ of $\A_{1P}$ we let $i^{-1}$ denote the most recently concatenated state. For example, for $i = i_{2,1} \circ \cdots \circ i_{2,t}$, $i^{-1} = i_{2,t}$\footnote{We assume that it is possible to separate a concatenation into states of $\A_{2P}$ after the fact. This assumption is easily (but less neatly) removed using a separator character $\bot$.}. We then define the internal algorithm of $\A_{1P}$ by $\A_{1P, \In}(i, x) = i \circ \A_{2P, \In}(i^{-1}, x)$. Finally, we define the output algorithm of $\A_{1P}$ by $\A_{1P, \Ou}(i) = \A_{2P, \Ou}(i^{-1})$. As a result, $\A_{1P, \Ou}(\A_{1P, \In}(s)) = \A_{2P, \Ou}(\A_{2P, \In}(s))$, and $\A_{1P}$ and $\A_{2P}$ have identical output distributions.
    
	We will prove this result for discrete state spaces; a similar approach works for continuous state spaces if we replace probability mass functions with densities. To prove $\eps$-pan-privacy of $\A_{1P}$ against one intrusion, it suffices to fix neighboring streams $s$ and $s'$, internal state set $i$, output state set $o$, stream position $t$, and show 
	$$\frac{\P{\A_{1P}}{\A_{1P, \In}(s_{\leq t}) = i} \P{\A_{1P}}{\A_{1P, \Ou}(\A_{1P,\In}(s)) = o \mid \A_{1P, \In}(s_{\leq t}) = i}}{\P{\A_{1P}}{\A_{1P, \In}(s_{\leq t}') = i} \P{\A_{1P}}{\A_{1P, \Ou}(\A_{1P,\In}(s')) = o \mid \A_{1P, \In}(s_{\leq t}') = i}} \leq e^\eps.$$
	First, by the definition of $\A_{1P}$, it suffices to show
	\begin{equation}
	\label{eq:pan_local}
		\frac{\P{\A_{1P}}{\A_{1P, \In}(s_{\leq t}) = i} \P{\A_{2P}}{\A_{2P, \Ou}(\A_{2P,\In}(s)) = o \mid \A_{2P, \In}(s_{\leq t}) = i^{-1}}}{\P{\A_{1P}}{\A_{1P, \In}(s_{\leq t}') = i} \P{\A_{2P}}{\A_{2P, \Ou}(\A_{2P,\In}(s')) =o \mid \A_{2P, \In}(s_{\leq t}') = i^{-1}}} \leq e^\eps.
	\end{equation}
	Suppose streams $s$ and $s'$ differ at time $t^*$, i.e. $s_{t^*} \neq s_{t^*}'$. If $t^* > t$, then we immediately have $\P{\A_{1P}}{\A_{1P,\In}(s_{\leq t}) = i} = \P{\A_{1P}}{\A_{1P,\In}(s_{\leq t}') = i}$, and $\tfrac{\P{\A_{2P}}{\A_{2P, \Ou}(\A_{2P,\In}(s)) = o \mid \A_{2P, \In}(s_{\leq t}) = i^{-1}}}{\P{\A_{2P}}{\A_{2P, \Ou}(\A_{2P,\In}(s')) =o \mid \A_{2P, \In}(s_{\leq t}') = i^{-1}}} \leq e^\eps$ follows from the $\eps$-pan-privacy of $\A_{2P}$. Thus Inequality~\ref{eq:pan_local} holds.
	
	The remaining case is when $t^* \leq t$. Here, $\tfrac{\P{\A_{2P}}{\A_{2P, \Ou}(\A_{2P,\In}(s)) = o \mid \A_{2P, \In}(s_{\leq t}) = i^{-1}}}{\P{\A_{2P}}{\A_{2P, \Ou}(\A_{2P,\In}(s')) =o \mid \A_{2P, \In}(s_{\leq t}') = i^{-1}}} = 1$, and we need to upper bound $\tfrac{\P{\A_{1P}}{\A_{1P, \In}(s_{\leq t}) = i}}{\P{\A_{1P}}{\A_{1P, \In}(s_{\leq t}') = i}}$. Since $\A_{1P, \In}(s_{\leq t})$ is conditionally independent of $\A_{1P, \In}(s_{\leq t^*-1})$ given $\A_{1P, \In}(s_{\leq t^*})$, it suffices to show that $\tfrac{\P{\A_{1P}}{\A_{1P, \In}(s_{\leq t^*}) = i}}{\P{\A_{1P}}{\A_{1P, \In}(s_{\leq t^*}') = i}} \leq e^\eps$. Recall that $I_{j,t}$ is the random variable for the internal state of $\A_j$ after seeing the $t^{th}$ stream element. Then it is equivalent to show $\tfrac{\P{\A_{1P}}{I_{1,t^*} = i \mid S_{\leq t*} = s_{\leq t*}}}{\P{\A_{1P}}{I_{1,t^*} = i \mid S_{\leq t^*} = s_{\leq t*}'}} \leq e^\eps$.
        
        We introduce some additional notation to prove this claim. $i$ is an internal state for $\A_{1P}$ and is therefore a concatenation of internal states for $\A_{2P}$. Let $i_a$ denote the $a^{th}$ state in the concatenation $i$, and let $i_{a:b} = i_a \circ i_{a+1} \circ \cdots \circ i_b$, the concatenation of states $i_a$ through $i_b$. Then $\tfrac{\P{\A_{1P}}{I_{1,t^*} = i \mid S_{\leq t^*} = s_{\leq t^*}}}{\P{\A_{1P}}{I_{1,t^*} = i \mid S_{\leq t^*} = s_{\leq t^*}'}}$
        \begin{align*}
        =&\; \frac{\P{\A_{1P}}{I_{1,t^*-1} = i_{1:t^*-1} \mid S_{\leq t^*} = s_{\leq t^*}} \cdot \P{\A_{2P}}{I_{2,t^*} = i_{t^*} \mid S_{\leq t^*} = s_{\leq t^*}, I_{2,t^*-1} = i_{t^*-1}}}{\P{\A_{1P}}{I_{1,t^*-1} = i_{1:t^*-1} \mid S_{\leq t^*} = s_{\leq t^*}'} \cdot \P{\A_{2P}}{I_{2,t^*} = i_{t^*} \mid S_{\leq t^*} = s_{\leq t^*}', I_{2,t^*-1} = i_{t^*-1}}} \\
            =&\; \frac{\P{\A_{2P}}{I_{2,t^*} = i_{t^*} \mid S_{\leq t^*} = s_{\leq t^*}, I_{2,t^*-1} = i_{t^*-1}}}{\P{\A_{2P}}{I_{2,t^*} = i_{t^*} \mid S_{\leq t^*} = s_{\leq t^*}', I_{2,t^*-1} = i_{t^*-1}}} \\
            =&\; \frac{\P{\A_{2P}}{I_{2,t^*} = i_{t^*} \mid S_{t^*} = s_{t^*}, I_{2,t^*-1} = i_{t^*-1}}}{\P{\A_{2P}}{I_{2,t^*} = i_{t^*} \mid S_{t^*} = s_{t^*}', I_{2,t^*-1} = i_{t^*-1}}}
        \end{align*}
where the second equality uses the fact that $s_{<t^*} = s_{<t^*}'$, and the third equality uses $I_{2,t^*}$'s conditional independence from $S_{\leq t^*-1}$ given $I_{2,t^*-1}$. Now, since $I_{2, t^*-1}$ and $S_{t^*}$ are independent, we multiply by $1 = \tfrac{\P{\A_{2P}}{I_{2,t^*-1} = i_{t^*-1} \mid S_{t^*} = s_{t^*}}}{\P{\A_{2P}}{I_{2,t^*-1} = i_{t^*-1} \mid S_{t^*} = s_{t^*}'}}$ to get
        \begin{equation*}
            \frac{\P{\A_{2P}}{I_{2,t^*} = i_{t^*} \mid S_{t^*} = s_{t^*}, I_{2,t^*-1} = i_{t^*-1}}}{\P{\A_{2P}}{I_{2,t^*} = i_{t^*} \mid S_{t^*} = s_{t^*}', I_{2,t^*-1} = i_{t^*-1}}} = \frac{\P{\A_{2P}}{I_{2,t^*} = i_{t^*}, I_{2,t^*-1} = i_{t^*-1} \mid S_{t^*} = s_{t^*}}}{\P{\A_{2P}}{I_{2,t^*} = i_{t^*}, I_{2,t^*-1} = i_{t^*-1} \mid S_{t^*} = s_{t^*}'}} \leq e^\eps
        \end{equation*}
since $\A_{2P}$ is $\eps$-pan-private against two intrusions.
    \end{proof}
    
    Next, we show how to convert this pan-private algorithm $\A_{1P}$ into an equivalent locally private algorithm $\A_L$.
    
    \begin{lemma}
    \label{lem:one_pp_ldp}
        Let $\A_{1P}$ be an $\eps$-pan-private algorithm as described in Lemma~\ref{lem:two_one_pp}. Then there exists a sequentially interactive $\eps$-locally private algorithm $\A_L$ whose transcript distribution $\Pi_t$ is identical to the $\A_{1P}$'s state distribution $I_t$ at each time $t$.
    \end{lemma}
    \begin{proof}
        At each time $t$, $\A_{1P}$ computes a function $\A_{1P}(i_{t-1},s_t)$ of its current state and the current element in the stream and concatenates it to its current state. We define $\A_L$ to use $\A_{1P}(i_{t-1}, \cdot)$ as a randomizer, add the result $\A_{1P}(i_{t-1},s_t)$ to the transcript, and continue.
        
        $\A_L$ is sequentially interactive because we take a single pass through the stream. Furthermore, because $\A_{1P}$ is $\eps$-pan-private and maintains all previous states, the transcript $\Pi_t$ of $\A_L$ is an $\eps$-differentially private function of the user data. Thus $\A_L$ is $\eps$-locally private. Finally, recalling that Definition~\ref{def:transcript} defined a transcript to record not only outputs but the randomizers used as well, let $\Pi_t^{-R}$ denote $\Pi_t$ with the randomizers omitted. Then for any input stream $s$, $\Pi_t^{-R}$ is distributed identically to $I_t$.
    \end{proof}
    
    We now combine Lemma~\ref{lem:two_one_pp} and Lemma~\ref{lem:one_pp_ldp}: any $\A_{2P}$ that is $\eps$-pan-private against two intrusions yields a sequentially interactive $\eps$-locally private $\A_L$ such that for any input stream $s$ and time $t$, $I_{2,t}$ is distributed identically to $\Pi_t^{-R, -1}$, the most recent addition to the transcript.
    
    \underline{$\Leftarrow$ (local to pan)}: Let $\A_L \colon \Pi \to R$ be a sequentially interactive $\eps$-locally private protocol mapping transcripts to randomizers, and let $\A_{\In} \colon \In \times \X \to \In$ be $\A_{2P}$'s internal algorithm with initial state $\emptyset$. We define $\A_{\In}(\emptyset,x_1) = (\emptyset, \A_L(\emptyset), \A_L(\emptyset)(x_1))$ and define other internal states $i$ by $\A_{\In}(i, x) = i \circ (\A_L(i), \A_L(i)(x))$, the concatenation of the existing state $i$ and the (randomizer, output) pair $(\A_L(i), \A_L(i)(x))$. Thus $I_t = \Pi_t$ at each time $t$. Finally, we define the output algorithm to be the identity function $\A_{\Ou}(i) = i$.
    
     Since $\A_L$ is $\eps$-locally private, its final transcript $\Pi$ is an $\eps$-differentially private function of the stream: for any transcript realization $\pi$ and neighboring streams $s$ and $s'$, $\tfrac{\P{\A_L}{\Pi = \pi \mid S =s}}{\P{\A_L}{\Pi = \pi \mid S = s'}} \leq e^{\eps}$. Letting $I^*$ be a random variable for the final internal state of $\A_{2P}$, it follows that $\tfrac{\P{\A_{2P}}{I^* = \pi \mid S =s}}{\P{\A_{2P}}{I^* = \pi \mid S = s'}} \leq e^{\eps}$. Thus the final internal state $I$ of $\A_{2P}$ is also an $\eps$-differentially private function of the stream. Moreover, because it is a transcript, $I^*$ includes a record of all previous internal states. Thus the additional view of any two internal states (in fact, any number of internal states) is still an $\eps$-differentially private function of the stream: fixing times $t_1, \ldots, t_c$ and corresponding internal states $\pi_1, \ldots, \pi_c$, 
     $$\frac{\P{\A_{2P}}{I_{t_1} = \pi_1, \ldots, I_{t_c} = \pi_c, I^* = i \mid S = s}}{\P{\A_{2P}}{I_{t_1} = \pi_1, \ldots, I_{t_c} = \pi_c, I^* = i \mid S = s'}} \leq e^{\eps}.$$
     Finally, since the output of $\A_{2P}$ is the final state $I^*$, $\A_{2P}$ is $\eps$-pan-private against arbitrarily many (and, in particular, two) intrusions.
\end{proof}

\section{Uniformity Testing}
\label{sec:uni}
We now turn to upper bounds for pan-privacy against a single intrusion. Our benchmark problem is \emph{uniformity testing}. In uniformity testing, a tester receives i.i.d. sample access to an unknown discrete distribution $p$ over $[k]$ and must determine with nontrivial constant probability whether $p$ is uniform or $\alpha$-far from uniform in total variation distance. Below, let $U_k$ denote the uniform distribution over $[k]$.

\begin{definition}[Uniformity testing]
\label{def:uni} An algorithm $\A$ is a \emph{uniformity tester on $m$ samples} if, given $m$ i.i.d. samples from $p$, 
\begin{enumerate}
	\item when $p = U_k$, with probability $\geq 2/3$ $\A$ outputs ``uniform'', and
	\item when $\tv{p}{U_k} \geq \alpha$, with probability $\geq 2/3$ $\A$ outputs ``non-uniform''.
\end{enumerate}
\end{definition}

The specific choice of $2/3$ is arbitrary. The important point is that there is a constant separation between output probabilities, which can be amplified to $2/3$ with a constant number of repetitions. We therefore focus on achieving any such constant separation. Details for this standard perspective appear in Appendix~\ref{sec:sep_details}.

\subsection{Warmup: \ppu}
We start with a suboptimal uniformity tester \ppu . \ppu~is a warmup and eventual building block for a better algorithm \ppub~(Section~\ref{subsec:better}).

Like many uniformity testers,  \ppu~computes a statistic on the data and compares it to a threshold. The statistic is designed to be small when $p$ is uniform and large if $p$ is $\alpha$-far from uniform. For \ppu , our statistic is
$$Z' = \sum_{i=1}^k \frac{(H_i - m/k)^2 - H_i}{m/k}$$
where $m$ is the number of samples and $H$ is a noisy histogram over $[k]$ where bin $i$ counts the number of occurrences of element $i$ in the stream. $H$ contains Laplace noise added to each bin both before and after the stream. The first addition of noise ensures the privacy of the internal states during the stream, while the second addition of noise is for the privacy of the final output. Pseudocode for \ppu~appears below; values for $m$ and $T_U$ are determined in the proof of Lemma~\ref{lem:Z'}.

\begin{algorithm}
\caption{Pan-private uniformity tester \ppu} 
\begin{algorithmic}[H]
    \Require{privacy parameter $\eps$, domain $[k]$}
    \State Set sample size $m' \sim \poi{m}$ and threshold $T_U$
    \State Initialize private histogram $H \gets \lap{\tfrac{1}{\eps}}^k \in \mathbb{R}^k$
    \For{stream elements $s_t = s_1, \ldots, s_{m'}$}
        \State $H_{s_t} \gets H_{s_t} + 1$   
    \EndFor
    \State $H \gets H + \lap{\tfrac{1}{\eps}}^k \in \mathbb{R}^k$
    \State $Z' \gets \sum_{i=1}^k \frac{(H_i - m/k)^2 - H_i}{m/k}$
    \If{$Z' > T_U$}
    		\State Output ``non-uniform''
    	\Else
    		\State Output ``uniform''
    \EndIf
\end{algorithmic}
\end{algorithm}

Inspired by similar statistics in non-private testing~\cite{AJOS13, CDVV14, ADK15},~\citet{CDK17} originally studied $Z'$ for centrally private identity testing. However, they lower bounded its variance and argued that high variance makes it a suboptimal centrally private tester. We instead upper bound its variance and show that $Z'$ yields a nontrivial pan-private uniformity tester.

Our argument is simple. First, we upper bound the variance of $Z'$. We then apply Chebyshev's inequality to upper bound $Z'$ when $p$ is uniform and lower bound $Z'$ when $p$ is $\alpha$-far from uniform. These bounds drive our choice of the threshold $T_U$. We then compute the number of samples $m$ required to separate these quantities on either side of $T_U$. Since the proof largely consists of straightforward calculations, we defer it to Section~\ref{sec:app_b} in the Appendix.

\begin{lemma}
\label{lem:Z'}
	For $m = \Omega\left(\frac{k^{3/4}}{\alpha \eps} + \tfrac{\sqrt{k}}{\alpha^2}\right)$, \ppu~is an $\eps$-pan-private uniformity tester on $m$ samples.
\end{lemma}

Note from the pseudocode for \ppu~that we actually draw $m' \sim \poi{m}$ samples, not $m$. This ``Poissonization'' trick is important for the analysis used to prove Lemma~\ref{lem:adk}. Since $\poi{m}$ concentrates around $m$~\cite{C16}, a uniformity tester on $\poi{m}$ samples implies a uniformity tester on a constant factor more samples with a constant decrease in success probability (see Section D.4 in the survey of~\citet{C15} for a more detailed discussion of Poissonization). 

\subsection{Optimal pan-private tester: \ppub}
\label{subsec:better}
We now use \ppu~as a building block for a more complex tester \ppub. At a high level,  \ppub~splits the difference between local and central uniformity testers. We briefly recap these approaches for context.

Centrally private uniformity testers compute a fine-grained statistic depending on the empirical counts of each element $i \in [k]$. Specific methods include $\chi^2$-style statistics~\cite{CDK17}, collision-counting~\cite{ADR18}, and empirical total variation distance from $U_k$~\cite{ASZ18}, but all of these methods depend on accurate counts for each $i \in [k]$.~\citet{CDK17} observed that adding Laplace noise to each such count before analyzing the statistic is centrally private. The cost is a large decrease in accuracy. This is unfortunate in our pan-private setting, as pan-privacy appears to force the same kind of per-count noise. Intuitively, a pan-private tester might benefit by maintaining a coarser statistic  --- i.e., one that tracks fewer counts --- that is easier to maintain privately.

The best known\footnote{Note that existing lower bounds, including the one in this paper, have not ruled out the possibility that  a \emph{fully interactive} locally private uniformity tester obtains better sample complexity.} locally private uniformity tester, due to~\citet{ACFT19}, uses an extreme version of this coarser strategy. Their approach randomly halves the domain $[k]$ into sets $U$ and $U^c$ and compares the number of samples falling into each. They prove that if $p$ is sufficiently non-uniform to start, then $p(U)$ and $p(U^c)$ will also be non-uniform --- albeit to a much smaller degree --- with constant probability. This reduces uniformity testing to a simpler binary testing problem that, because of its much smaller domain, is more amenable to local privacy. However, it does so at the cost of a large reduction in testing distance, which makes the core distinguishing problem harder. Thus both locally private and pan-private versions of this approach have sample complexity $\Omega(k)$. Intuitively, because pan-privacy does not force as much noise as local privacy, a pan-private algorithm might benefit by maintaining a finer statistic.

\ppub~capitalizes on both of these ideas. First, it randomly partitions $[k]$ into $n$ groups $G_1, \ldots, G_n$ of size $\Theta(k/n)$. It then runs \ppu~to test uniformity of the induced distribution over $[n]$, treating samples falling in each $G_j$ as samples of $j \in [n]$. 

\ppub~thus intermediates between the central and local approaches. It chooses $n = n(\alpha, \eps, k)$ according to $\tfrac{k^{2/3}\eps^{4/3}}{\alpha^{4/3}}$. When $\tfrac{k^{2/3}\eps^{4/3}}{\alpha^{4/3}} < 2$, $n(\alpha, \eps, k) = 2$ and \ppub~uses the half-partition approach from local privacy. When $\tfrac{k^{2/3}\eps^{4/3}}{\alpha^{4/3}} > k$, then $n(\alpha, \eps, k) = k$ and \ppub~uses the unpartitioned approach from central privacy. Finally, when $\tfrac{k^{2/3}\eps^{4/3}}{\alpha^{4/3}} \in [2,k]$, $n(\alpha, \eps, k) = \lfloor \tfrac{k^{2/3}\eps^{4/3}}{\alpha^{4/3}} \rfloor$ and \ppub~takes a middle ground. These choices enable \ppub~to calibrate the noise contributed by privately maintaining different counts with the testing distance $\alpha$ Making this tradeoff work relies crucially on the $O\left(\tfrac{1}{\alpha}\right)$ dependence on distance achieved by \ppu~in its $k^{3/4}$ term. In contrast, the $\Omega\left(\tfrac{k}{\alpha^2}\right)$ dependence of the best known locally private uniformity tester yields no improvement with this approach. Pseudocode for \ppub~appears below.

\begin{algorithm}
\caption{Improved pan-private uniformity tester \ppub} 
\begin{algorithmic}[H]
    \Require{privacy parameter $\eps$, domain $[k]$}
    \If{$\frac{k^{2/3}\eps^{4/3}}{\alpha^{4/3}} < 2$}
    		\State $n \gets 2$
    \ElsIf{$\frac{k^{2/3}\eps^{4/3}}{\alpha^{4/3}} > k$}
    		\State $n \gets k$
    \Else
    		\State $n \gets \lfloor \frac{k^{2/3}\eps^{4/3}}{\alpha^{4/3}} \rfloor $
    \EndIf
    \State Randomly partition $[k]$ into $n$ groups $G_1, \ldots, G_n$ of size $\Theta(k/n)$
    \State Run $\ppu(\eps, [n])$, treating each element $s_t \in G_j$ as $j \in [n]$
\end{algorithmic}
\end{algorithm}

For this reduction to work, the aforementioned decrease in testing distance between $[k]$ and $[n]$ must not be too large. We show this in Lemma~\ref{lem:distance}, which generalizes a similar result of~\citet{ACFT19} for the special case of a partition into two subsets. As pointed out by a reviewer, this generalization is not new (see Theorem 3.2 from~\citet{ACHST19}), but we include a proof in Section~\ref{sec:app_b} of the Appendix for completeness.

\begin{lemma}
\label{lem:distance}
    Let $p$ be a distribution over $[k]$ such that $\tv{p}{U_k} = \alpha$ and let $G_1, \ldots, G_n$ be a uniformly random partition of $[k]$ into $n > 1$ subsets of size $\Theta(k/n)$. Define induced distribution $p_n$ over $[n]$ by $p_n(j) = \sum_{i \in G_j} p(i)$ for each $j \in [n]$. Then, with probability $\geq \tfrac{1}{954}$ over the selection of $G_1, \ldots, G_n$, $$\tv{p_n}{U_n} = \Omega\left(\alpha\sqrt{\tfrac{n}{k}}\right).$$
\end{lemma}

Due to the $1/954$ success probability of Lemma~\ref{lem:distance}, we have a smaller (but still constant) separation between output probabilities. We thus use the amplification argument discussed after Definition~\ref{def:uni} to get Theorem~\ref{thm:ppub}. The guarantee combines Lemma~\ref{lem:distance} with Lemma~\ref{lem:Z'}, substituting $n$ for $k$ and $\alpha\sqrt{\tfrac{n}{k}}$ for $\alpha$. 

\begin{theorem}
\label{thm:ppub}
     For $m = \Omega\left(\frac{k^{2/3}}{\alpha^{4/3}\eps^{2/3}} + \frac{\sqrt{k}}{\alpha^2} + \frac{\sqrt{k}}{\alpha \eps}\right)$, $\ppub$ is an $\eps$-pan-private uniformity tester on $m$ samples.
\end{theorem}
\begin{proof}
	\underline{Privacy}: \ppub~only interacts with the data through \ppu, so \ppub~inherits \ppu 's pan-privacy guarantee. 
	
	\underline{Sample complexity}: Substituting $n$ for $k$ and $\alpha\sqrt{\tfrac{n}{k}}$ for $\alpha$ in Lemma~\ref{lem:Z'}, we require 
	\begin{equation}
	\label{eq:new_m}
		m = \Omega\left(\frac{n^{1/4}\sqrt{k}}{\alpha \eps} + \frac{k}{\alpha^2\sqrt{n}}\right).
	\end{equation}
	We consider the three cases for $\frac{k^{2/3}\eps^{4/3}}{\alpha^{4/3}}$. Together, these cases exhaust the possible relationships among $\alpha, k$, and $\eps$, with a different highest-order term in each. This leads to the three terms in our bound.
	
	First, if $\frac{k^{2/3}\eps^{4/3}}{\alpha^{4/3}} \in [2,k]$, then $n = \lfloor \frac{k^{2/3}\eps^{4/3}}{\alpha^{4/3}} \rfloor$. By Equation~\ref{eq:new_m} it is enough for
	$$m = \Omega\left(\frac{k^{1/6}\eps^{1/3}\sqrt{k}}{\alpha^{1/3}\alpha\eps} + \frac{k}{\alpha^2 \cdot \tfrac{k^{1/3}\eps^{2/3}}{\alpha^{2/3}}}\right) = \Omega\left(\frac{k^{2/3}}{\alpha^{4/3}\eps^{2/3}}\right).$$
	
	Next, if $\frac{k^{2/3}\eps^{4/3}}{\alpha^{4/3}} > k$, then $n = k$, and Equation~\ref{eq:new_m} necessitates $m = \Omega\left(\tfrac{k^{3/4}}{\alpha \eps} + \tfrac{\sqrt{k}}{\alpha^2}\right)$. The condition $\frac{k^{2/3}\eps^{4/3}}{\alpha^{4/3}} > k$ gives $\tfrac{\eps^4}{\alpha^4} > k$, so $\tfrac{\eps}{\alpha} > k^{1/4}$, and then multiplying both sides by $\tfrac{\sqrt{k}}{\alpha \eps}$ gives $\tfrac{\sqrt{k}}{\alpha^2} > \tfrac{k^{3/4}}{\alpha \eps}$.Thus it suffices for $m = \Omega\left(\tfrac{\sqrt{k}}{\alpha^2}\right)$.
	
	Finally, if $\frac{k^{2/3}\eps^{4/3}}{\alpha^{4/3}} < 2$, then $n = 2$ and by Equation~\ref{eq:new_m} we require $m = \Omega\left(\tfrac{\sqrt{k}}{\alpha \eps} + \tfrac{k}{\alpha^2}\right)$. $\frac{k^{2/3}\eps^{4/3}}{\alpha^{4/3}} < 2$ implies $\eps < \tfrac{2\alpha}{\sqrt{k}}$, so multiplying both sides by $\tfrac{k}{\alpha^2\eps}$ yields $ \tfrac{k}{\alpha^2} < \tfrac{2\sqrt{k}}{\alpha \eps}$ and $\tfrac{\sqrt{k}}{\alpha \eps} = \Omega\left(\tfrac{k}{\alpha^2}\right)$. Thus it suffices for $m = \Omega\left(\tfrac{\sqrt{k}}{\alpha \eps}\right)$.
\end{proof}

\section{Lower Bounds}
\label{sec:uni_lower}
We now turn to lower bounds. Our first result gives a tight (in $k$) $\Omega\left(\tfrac{k^{2/3}}{\alpha^{4/3}\eps^{2/3}}\right)$ lower bound for $\eps$-pan-private testing (Section~\ref{sec:pp_lb}). Our second result extends the previous $\Omega\left(\tfrac{k}{\alpha^2\eps^2}\right)$ lower bound for noninteractive $(\eps,\delta)$-locally private uniformity testing (\cite{ACFT19}) to the sequentially interactive case (Section~\ref{sec:lp_lb}).

Both of our lower bounds adapt the approach used by~\citet{DGKR19} to prove testing lower bounds under memory restrictions and communication restrictions. Like~\citet{DGKR19}, we consider the problem of distinguishing between two distributions. If uniform random variable $X$ is 0 then the distribution is uniform. If $X$ is 1 then each element has probability mass slightly perturbed from uniform such that the distribution is $\alpha$-far from uniform in total variation distance. Our argument then proceeds by upper bounding the mutual information between the random variable $X$ and the algorithm's internal state (in the pan-private case) or transcript (in the locally private case). Controlling this quantity lower bounds the number of samples required to identify $X$. This gives the final uniformity testing sample complexity lower bounds.

The main difference in our lower bounds is that~\citet{DGKR19} restrict their algorithm to use an internal state with $b$ bits of memory. This memory restriction immediately implies that the internal state's entropy (and thus its mutual information with any other random variable) is also bounded by $b$. In our case, we must use our privacy restrictions to replace this result. Doing so constitutes the bulk of our arguments.

Finally, we note that these results add to lines of work conceptually connecting restricted memory to pan-privacy~\cite{DNPRY10, MMNW11} and connecting restricted communication to local privacy~\cite{MMPRTV10, ACFT19, DR19, ACHST19, JMR20}.

\subsection{Pan-private Lower Bound}
\label{sec:pp_lb}
We start with the pan-private lower bound. While we state our result using $\alpha \leq 1/2$, the choice of $1/2$ is arbitrary: the same argument works for any $\alpha$ bounded below 1 by a constant. A short primer on the information theory used in our argument appears in Appendix~\ref{sec:info}.

\begin{theorem}
\label{thm:pplb}
    For $\eps = O(1)$ and $\alpha \leq 1/2$, any $\eps$-pan-private uniformity tester requires $m = \Omega\left(\frac{k^{2/3}}{\alpha^{4/3}\eps^{2/3}} + \frac{\sqrt{k}}{\alpha^2} + \frac{1}{\alpha \eps}\right)$ samples.
\end{theorem}
\begin{proof}
		First, recall the centrally private lower bound~\cite{ASZ18}:
		$$m = \Omega\left(\frac{\sqrt{k}}{\alpha^2} + \frac{\sqrt{k}}{\alpha\sqrt{\eps}} + \frac{k^{1/3}}{\alpha^{4/3}\eps^{2/3}} + \frac{1}{\alpha\eps}\right).$$
		We will prove $m = \Omega\left(\frac{k^{2/3}}{\alpha^{4/3}\eps^{2/3}}\right)$ in the pan-private case. $\tfrac{k^{2/3}}{\alpha^{4/3}\eps^{2/3}}$ dominates the third term above and also dominates the second term for $\eps = O(1)$, so this produces our final lower bound.
		
        We start with the lower bound construction used by~\citet{DGKR19}, which itself uses the Paninski lower bound construction~\cite{P08}. Let $X$ be a uniform random bit determining which of two distributions over $[2k]$ generates the samples. For both $X = 0$ and $X = 1$ we draw $Y_1, \ldots, Y_k \in \{\pm 1\}$ i.i.d. uniformly at random. If $X = 0$, $p = U_{2k}$. If instead $X = 1$, then we pair the bins as $\{1,2\}, \{3,4\}, \ldots, \{2k-1, 2k\}$ and define $p(2j-1) = \tfrac{1 + Y_j\alpha}{2k}$ and $p(2j) = \tfrac{1 - Y_j\alpha}{2k}$. Thus if $X = 0$ then $p$ is uniform, and if $X = 1$ each pair $i$ of bins is biased toward one of the bins according to $Y_j$. Equivalently, we can view each sample $S_t \sim p$ as a pair $(J_t, V_t)$ where $J_t \in [k]$ determines the bin pair chosen and $V_t \in \{0,1\}$ determines which of the bin pair is chosen. Thus $J_t \sim U_k$, and $V_t \sim \ber{\tfrac{1}{2}}$ if $X = 0$ or $V_t \sim \ber{[1+\alpha Y_{j_t}]/2}$ if $X = 1$, where $\ber{\cdot}$ denotes the Bernoulli distribution.
        
        To avoid confusion with the mutual information $I(\cdot)$, denote by $M_t$ the random variable for the internal state of the algorithm after seeing sample $S_t$. Our goal is to upper bound the mutual information between $X$ and the internal state after $m$ samples,
    \begin{align*}
        I(X;M_m) =&\; \sum_{t=1}^{m} I(X; M_t) - I(X; M_{t-1}) \\
        \leq&\; \sum_{t=1}^{m} I(X; M_{t-1}, S_{t}) - I(X; M_{t-1}) \\
        =&\; \sum_{t=1}^{m} I(X; S_{t} \mid M_{t-1}) \\
        =&\; \sum_{t=1}^{m} I(X; V_{t} \mid M_{t-1}, J_{t}) \stepcounter{equation}\tag{\theequation}\label{eq:4}
    \end{align*}
    where the last equality uses $S_{t} = (J_{t}, V_{t})$ and the independence of $X$ and $J_{t}$
    
    We now have a narrower goal: we choose an arbitrary term in the sum in Equation~(\ref{eq:4}) and upper bound it. For neatness, we use the convention that $H_2(p)$ is the entropy of a $\ber{p}$ random variable. When subscripting we abuse notation and let $a \sim A$ denote a sample $a$ from the distribution for random variable $A$. The following reproduces (and slightly expands) the first part of the argument given by~\citet{DGKR19}. It largely reduces to rewriting mutual information in terms of binary entropy and expanding conditional probabilities.
    
    We start by rewriting the chosen term $I(X; V_t \mid M_{t-1}, J_t)$ as
    \begin{align*}
     =&\; \E{m^* \sim M_{t-1}}{\E{j \sim J_t}{H(V_t \mid M_{t-1} = m^*, J_t = j)}} \\
        &-\; \E{m^* \sim M_{t-1}}{\E{j \sim J_t}{\E{x \sim X}{H(V_t \mid M_{t-1} = m^*, J_t = j, X = x)}}} \\ 
        =&\; \E{m^* \sim M_{t-1}}{\E{j \sim J_t}{H_2(\P{}{V_t = 0 \mid M_{t-1} = m^*, J_t = j})}} \\
        &-\; \E{m^* \sim M_{t-1}}{\E{j \sim J_t}{\P{}{X = 1 \mid M_{t-1} = m^*, J_t = j}H_2(\P{}{V_t = 0 \mid M_{t-1} = m^*, J_t = j, X = 1})}} \\
        &-\; \E{m^* \sim M_{t-1}}{\E{j \sim J_t}{\P{}{X = 0 \mid M_{t-1} = m^*, J_t = j}H_2(\P{}{V_t = 0 \mid M_{t-1} = m^*, J_t = j, X = 0})}}
    \end{align*}
    where the second equality uses $H_2(p) = H_2(1-p)$. Let $\beta_{t-1}^{m^*,j} = \P{}{X = 1 \mid M_{t-1} = m^*, J_t = j}$. Since $J_t$ is a uniform draw from $[k]$ independent of $M_{t-1}$, we now continue the above chain of equalities as
    \begin{align*}
    		=&\; \E{m^* \sim M_{t-1}}{\frac{1}{k}\sum_{j=1}^k H_2\left(\P{}{V_t = 0 \mid M_{t-1} = m^*, J_t = j}\right)} \\
    		&-\; \E{m^* \sim M_{t-1}}{\frac{1}{k} \sum_{j=1}^k \beta_{t-1}^{m^*,j} H_2\left(\P{}{V_t = 0 \mid M_{t-1}=m^*, J_t = j, X = 1}\right)} \\
    		&-\; \E{m^* \sim M_{t-1}}{\frac{1}{k} \sum_{j=1}^k (1-\beta_{t-1}^{m^*,j}) H_2\left(\P{}{V_t = 0 \mid M_{t-1}=m^*, J_t = j, X = 0}\right)}
    		 \stepcounter{equation}\tag{\theequation}\label{eq:entropies_1}.
    	\end{align*}
    	Now recall that $V_t \sim \ber{\tfrac{1}{2}}$ when $X = 0$ and $V_t \sim \ber{[1 + \alpha Y_{J_t}]/2}$ when $X=1$. Then we can rewrite $\P{}{V_t = 0 \mid M_{t-1} = m^*, J_t = j}$ as
    	\begin{align*}
    		=&\; \beta_{t-1}^{m^*,j} \P{}{V_t = 0 \mid X = 1, M_{t-1} = m^*, J_t = j} + (1-\beta_{t-1}^{m^*,j} )\P{}{V_t = 0 \mid X = 0, M_{t-1} = m^*, J_t = j} \\
    		=&\; \beta_{t-1}^{m^*,j} \P{}{V_t = 0 \mid X = 1, M_{t-1} = m^*, J_t = j, Y_j = 1}\P{}{Y_j = 1 \mid M_{t-1} = m^*} \\
    		&+\; \beta_{t-1}^{m^*,j} \P{}{V_t = 0 \mid X = 1, M_{t-1} = m^*, J_t = j, Y_j = -1}\P{}{Y_j = -1 \mid M_{t-1} = m^*} \\
    		&+\; (1-\beta_{t-1}^{m^*,j} )\P{}{V_t = 0 \mid X = 0} \\
    		=&\; \beta_{t-1}^{m^*,j}  \left(\P{}{Y_j = 1 \mid M_{t-1} = m^*} \cdot \frac{1 - \alpha}{2} + \P{}{Y_j = - 1 \mid M_{t-1} = m^*} \cdot \frac{1 + \alpha}{2}\right) + \frac{1 - \beta_{t-1}^{m^*,j} }{2} \\
    		=&\; \beta_{t-1}^{m^*,j} \E{}{\frac{1 - \alpha Y_j}{2} \mid M_{t-1} = m^*} + \frac{1-\beta_{t-1}^{m^*,j} }{2} \\
    		=&\; \frac{\beta_{t-1}^{m^*,j} (1 - \alpha \E{}{Y_j \mid M_{t-1} = m^*}}{2} + \frac{1-\beta_{t-1}^{m^*,j} }{2} = \frac{1 - \alpha \beta_{t-1}^{m^*,j}\E{}{Y_j \mid M_{t-1} = m^*}}{2}.
    	\end{align*}
    	where the first equality uses the independence of $Y_j$ from $X$ and $J_t$ as well as the independence of $V_t$ from $M_{t-1}$ and  $J_t$ conditioned on $X = 0$, and the second equality uses the independence of $V_t$ and $M_{t-1}$ conditioned on $X, J_t = j,$ and  $Y_j$. Thus
    	$$\P{}{V_t = 0 \mid M_{t-1} = m^*, J_t = j} = \frac{1 - \alpha \beta_{t-1}^{m^*,j}\E{}{Y_j \mid M_{t-1} = m^*}}{2}.$$
    Using the work above, we can also rewrite
    	$$\P{}{V_t = 0 \mid M_{t-1} = m^*, J_t = j, X = 1} = \frac{1 - \alpha \E{}{Y_j \mid M_{t-1} = m^*}}{2}$$
    	and
    	$$\P{}{V_t = 0 \mid M_{t-1} = m^*, J_t = j, X = 0} = \frac{1}{2}.$$
    	In the following chain of equalities, for space we let $E$ be the event that $M_{t-1} = m^*$. Now we can return to Equation~\ref{eq:entropies_1} and, since $H_2(\tfrac{1}{2}) = 1$, get
    	\begin{align*}
    		(\ref{eq:entropies_1}) =&\; \E{m^* \sim M_{t-1}}{\frac{1}{k} \sum_{j=1}^k \left(H_2\left(\frac{1 - \alpha \beta_{t-1}^{m^*,j}\E{}{Y_j \mid E}}{2}\right) - \beta_{t-1}^{m^*,j}H_2\left(\frac{1 - \alpha \E{}{Y_j \mid E}}{2}\right) - (1-\beta_{t-1}^{m^*,j})\right)} \\
    		=&\; \E{m^* \sim M_{t-1}}{\frac{1}{k} \sum_{j=1}^k \left(\beta_{t-1}^{m^*,j}\left[1 - H_2\left(\frac{1 - \alpha \E{}{Y_j \mid E}}{2}\right)\right] - \left[1 - H_2\left(\frac{1 - \alpha \beta_{t-1}^{m^*,j} \E{}{Y_j \mid E}}{2}\right)\right]\right)} \\
    		\leq&\; \E{m^* \sim M_{t-1}}{\frac{1}{k} \sum_{j=1}^k \left[1 - H_2\left(\frac{1 - \alpha \E{}{Y_j \mid E}}{2}\right)\right]} \\
    		=&\; \E{m^* \sim M_{t-1}}{\frac{1}{k} \sum_{j=1}^k \left[1 - H_2\left(\frac{1 + \alpha \E{}{Y_j \mid E}}{2}\right)\right]}
    		\stepcounter{equation}\tag{\theequation}\label{eq:entropies_2}
    	\end{align*}
    	where the inequality uses $H_2, \beta_{t-1}^{m^*,j} \leq 1$ and the equality uses $H_2\left(\frac{1}{2} - b\right) = H_2\left(\frac{1}{2} + b\right)$. We now control the terms with $H_2$. The Taylor series for $H_2(p)$ near $1/2$ is $H_2(p) = 1 - \tfrac{1}{2\ln(2)}\sum_{n=1}^\infty \tfrac{(1-2p)^{2n}}{n(2n-1)}$, so for $a < 1/2$
    $$1 - H_2\left(\frac{1}{2} + a\right) < \sum_{n=1}^\infty \frac{(2a)^{2n}}{n^2} = 4a^2\sum_{n=1}^\infty \frac{(2a)^{2n-2}}{n^2} < 4a^2\sum_{n=1}^\infty \frac{1}{n^2} = \frac{2a^2\pi^2}{3}.$$
    Substituting $1 - H_2\left(\frac{1}{2} + a\right) < \frac{2\pi^2a^2}{3}$ into Inequality~\ref{eq:entropies_2} and tracing back to Equation~\ref{eq:4},
    \begin{align*}
        I(X; V_t \mid M_{t-1}, J_t) <&\; \frac{\pi^2\alpha^2}{6k}\E{m^* \sim M_{t-1}}{\sum_{j=1}^k\E{}{Y_j \mid M_{t-1} = m^*}^2}
        \stepcounter{equation}\tag{\theequation}\label{eq:branch}
    \end{align*}

	We now depart from the argument of~\citet{DGKR19}. Our new goal is to upper bound 
    \begin{align*}
        A =&\; \E{m^* \sim M_{t-1}}{\sum_{j=1}^k\E{}{Y_j \mid M_{t-1} = m^*}^2} \\
        =&\;\E{m^* \sim M_{t-1}}{\sum_{j=1}^k \left(2\P{}{Y_j = 1 \mid M_{t-1} = m^*} - 1\right)^2} \\
        =&\; \E{m^* \sim M_{t-1}}{\sum_{j=1}^k \left(\frac{\P{}{M_{t-1} = m^* \mid Y_j = 1}}{\P{}{M_{t-1} = m^*}} - 1\right)^2}
    \end{align*}
    by Bayes' rule and $\P{}{Y_j = 1} = 1/2$. To upper bound this sum, we choose an arbitrary $j$ and show that $\tfrac{\P{}{M_{t-1} = m^* \mid Y_j = 1}}{\P{}{M_{t-1} = m^*}}$ is close to 1. We pause to recap what we've accomplished and what remains. Note that proving $\tfrac{\P{}{M_{t-1} = m^* \mid Y_j = 1}}{\P{}{M_{t-1} = m^*}} \approx 1$ ``looks like'' a privacy statement: we are claiming that the state distribution $M_{t-1}$ looks similar when its input distribution is slightly different. However, there is still a gap between a difference in input distribution and a difference in input. We close this gap in the following lemma, which relies on pan-privacy. 
    
    \begin{lemma}
    \label{lem:near_one}
        $\left|\frac{\P{}{M_{t-1} = m^* \mid Y_j = 1}}{\P{}{M_{t-1} = m^*}} - 1\right| = O\left(\frac{\alpha\eps t}{k}\right).$
    \end{lemma}
    \begin{proof}
    We will prove this claim by showing that both the numerator and denominator of $\frac{\P{}{M_{t-1} = m^* \mid Y_j = 1}}{\P{}{M_{t-1} = m^*}}$ fall into a bounded range. This implies that the whole fraction is near $1$.
    
    First consider the case $X=0$. Then the $Y_j$ are irrelevant, so $\left|\tfrac{\P{}{M_{t-1} = m^* \mid Y_j = 1}}{\P{}{M_{t-1} = m^*}} - 1\right| = 0$.
    
    Next, consider the case $X=1$. It will be useful to consider an equivalent method of sampling the stream $S$. At each time step $t$, we first sample a bin pair $J_t \sim_U [k]$ uniformly at random from the $k$ bin pairs. Having sampled bin pair $j$, with probability $1-\alpha$ we take a uniform random draw from $\{2j-1, 2j\}$. With the remaining probability $\alpha$, if $Y_j = 1$ then we sample $2j-1$, and if $Y_j = -1$ then we sample $2j$. Note that this method is equivalent because if $Y_j = 1$ then $\P{}{\text{sample } 2j-1} = \tfrac{1}{k} \cdot \tfrac{1-\alpha}{2} + \tfrac{\alpha}{k} = \tfrac{1+\alpha}{2k}$ and $\P{}{\text{sample } 2j} = \frac{1-\alpha}{2k}$, with these equalities swapped for $Y_j = -1$. With this view of sampling, let $E_{j,t}^{\alpha} = 1$ if $J_t = j$ and we sample from the $\alpha$ mixture component and $E_{j,t}^{\alpha} = 0$ otherwise. Finally, let $N_{j,t}^{\alpha} = \sum_{t'=1}^t E_{j,t'}^{\alpha}$, the number of samples from the $\alpha$ mixture component of bin pair $j$ through the first $t$ stream elements.
    
    We pause to justify bothering with this alternate view. We use it because the original ratio $\tfrac{\P{}{M_{t-1} = m^* \mid Y_j = 1}}{\P{}{M_{t-1} = m^*}}$ is comparing the views of $M_{t-1}$ depending on $Y_j$. It is not obvious how to directly use pan-privacy to reason about this comparison because $Y_j$ is a property of the distribution generating the samples (stream elements) rather than the samples themselves. In contrast, pan-privacy is a guarantee formulated in terms of the samples. By defining the $E_{j,t}^\alpha$ and $N_{j,t}^\alpha$ above we better connect $Y_j$ to the actual samples received. The alternate view therefore makes using pan-privacy easier.
    
        We first analyze the denominator of $\tfrac{\P{}{M_{t-1} = m^* \mid Y_j = 1}}{\P{}{M_{t-1} = m^*}}$. We can rewrite it as
        \begin{equation*}
        	\P{}{M_{t-1} = m^*} = \sum_{q=0}^{t-1} \P{}{M_{t-1} = m^* \mid N_{j,t-1}^{\alpha} = q} \cdot \P{}{N_{j,t-1}^{\alpha} = q}.
        	 \stepcounter{equation}\tag{\theequation}\label{eq:12}
        \end{equation*}
        Fix some $q  \in \{0, 1, \ldots, t-1\}$. Let $S_{j, \leq t^*}$ be the random variable for the bin pairs and component of $j$ sampled through time $t^*$, i.e. $S_{j,\leq t^*} = \{(J_{t}, E_{j, t}^{\alpha})\}_{t=1}^{t^*}$. Note that this means the tuple $(j',1)$ is possible only when $j' = j$. Define $\S_{j,q,t}^{\alpha}$ to be the set of realizations of $S_{j, \leq t}$ with exactly $q$ samples from the $\alpha$ component of bin pair $j$. Then
        \begin{align*}
            \P{}{M_{t-1} = m^* \mid N_{j,t-1}^{\alpha} = q} =&\; \sum_{s \in \S_{j,q,t-1}^{\alpha}} \P{}{M_{t-1}= m^* \mid S_{j, \leq t-1} = s} \cdot \P{}{S_{j, \leq t-1} = s \mid N_{j,t-1}^{\alpha} = q} \\
            =&\; \sum_{s \in \S_{j,q,t-1}^{\alpha}} \frac{1}{\binom{t-1}{q}k^{t-1-q}} \cdot \P{}{M_{t-1} = m^* \mid S_{j, \leq t-1} = s}
            \stepcounter{equation}\tag{\theequation}\label{eq:5}
        \end{align*}
        where the second equality uses the fact that, conditioned on $N_{j,t-1}^{\alpha} = q$, there are $\binom{t-1}{q}k^{t-1-q}$ equiprobable realizations of $S_{j, \leq t-1}$. Note that we are now reasoning directly about the stream's effect on the state $M_{t-1}$. This is much closer to the application of pan-privacy that we set out to achieve.
        
        Consider a length-$(t-1)$ realization $s \in \S_{j,q,t-1}^{\alpha}$. Recall that each index of $s$ takes one of $j+1$ possible values: $(1,0), (2,0), \ldots, (k,0)$, or $(j,1)$. Let $s' \in S_{j,0,t-1}^{\alpha}$ be a realization such that the Hamming distance $d_H(s,s') = q$, i.e. $s$ and $s'$ differ in exactly $q$ indices. Then because $M_{t-1}$ is an $\eps$-differentially private function of the stream, by group privacy (see e.g. Theorem 2.2 in the survey of Dwork and Roth~\cite{DR14})
        $$\P{}{M_{t-1}= m^* \mid S_{j, \leq t-1} = s} \leq e^{q\eps} \P{}{M_{t-1} = m^* \mid S_{j, \leq t-1} = s'}.$$
        Moreover, there are exactly $k^q$ such $s'$ for each such $s$. Denote this set of $s'$ by $T_{s,q}$. We can now continue
        \begin{align*}
            (\ref{eq:5}) =&\; \sum_{s \in \S_{j,q,t-1}^{\alpha}} \frac{1}{k^q} \sum_{s' \in T_{s,q}} \frac{1}{\binom{t-1}{q}k^{t-1-q}} \cdot \P{}{M_{t-1} = m^* \mid S_{j, \leq t-1} = s} \\
            \leq&\; \sum_{s \in \S_{j,q,t-1}^{\alpha}} \sum_{s' \in T_{s,q}} \frac{e^{q\eps}}{\binom{t-1}{q}k^{t-1}} \cdot \P{}{M_{t-1} = m^* \mid S_{j,\leq t-1} = s'} \\
            =&\; \sum_{s' \in S_{j,0,t-1}^{\alpha}} \frac{e^{q\eps}}{k^{t-1}} \cdot \P{}{M_{t-1} = m^* \mid S_{j, \leq t-1} = s'} \\
            =&\; \sum_{s' \in S_{j,0,t-1}^{\alpha}} e^{q\eps} \cdot \P{}{M_{t-1} = m^* \mid S_{j, \leq t-1} = s'} \cdot \P{}{S_{j, \leq t-1} = s' \mid N_{j,t-1}^{\alpha} = 0} \\
            =&\; e^{q\eps} \P{}{M_{t-1} = m^* \mid N_{j,t-1}^{\alpha} = 0}
        \end{align*}
        where the first inequality uses the above group privacy guarantee; the second equality uses the fact that, for a given $s' \in T_{s,q}$, there are exactly $\binom{t-1}{q}$ length-$(t-1)$ realizations $s$ with $q$ samples from the $\alpha$ mixture component from bin pair $j$ and $d_H(s,s') = q$; and the last equality uses the fact that $M_{t-1}$ and $N_{j,t-1}^{\alpha}$ are independent conditioned on $S_{j, \leq t-1}$. Note that this expression depending only on the conditioning for $N_{j,t-1}^\alpha = 0$ is useful because it will give us a ``fixed point'' to relate the numerator and denominator analyses. By expressing both quantities with respect to this condition, we can better compare them (and in particular, obtain a cancellation in the final ratio).
        
         Returning to Equation~\ref{eq:12}
        $$\P{}{M_{t-1} = m^*} = \sum_{q=0}^{t-1} \P{}{M_{t-1} = m^* \mid N_{j,t-1}^{\alpha} = q} \cdot \P{}{N_{j,t-1}^{\alpha} = q}$$
        we get
        \begin{align*}
            \P{}{M_{t-1} = m^*} \leq&\; \sum_{q=0}^{t-1} e^{q\eps} \P{}{M_{t-1} = m^* \mid N_{j,t-1}^{\alpha} = 0} \cdot \P{}{N_{j,t-1}^{\alpha} = q} \\
            =&\; \P{}{M_{t-1} = m^* \mid N_{j,t-1}^{\alpha} = 0} \cdot \sum_{q=0}^{t-1} e^{q\eps}\P{}{N_{j,t-1}^{\alpha} = q} \\
            =&\; \P{}{M_{t-1} = m^* \mid N_{j,t-1}^{\alpha} = 0} \cdot \E{}{e^{\eps N_{j,t-1}^{\alpha}}}.
            \stepcounter{equation}\tag{\theequation}\label{eq:8}
        \end{align*}
        To analyze this last quantity, recall that we defined random variable $E_{j,t}^{\alpha}$ as the indicator variable for drawing stream element $t$ from the $\alpha$ mixture component of bin pair $j$. Then
        $$\E{}{e^{\eps N_{j,t-1}}} = \E{}{e^{\sum_{i=1}^{t-1} \eps E_{j,i}^{\alpha}}} = \prod_{i=1}^{t-1} \E{}{e^{\eps E_{j,i}^{\alpha}}} = \left[\left(1 - \frac{\alpha}{k}\right)e^0 + \frac{\alpha}{k}e^\eps\right]^{t-1} = \left[1 + \frac{\alpha(e^\eps-1)}{k}\right]^{t-1}.$$
        Since $1+x \leq e^x$, $[1 + \tfrac{\alpha(e^\eps-1)}{k}]^{t-1} \leq e^{\tfrac{\alpha(e^\eps-1)(t-1)}{k}}$. We analyze this quantity in cases.
        
        In the first case, $\tfrac{\alpha(e^\eps-1)(t-1)}{k} \geq 1$. Then $t > \tfrac{k}{\alpha(e^\eps-1)}$, and since $\eps = O(1)$ there exists constant $C$ such that $t > C \frac{k}{\alpha \eps}$. $t \leq m$ so $m > C\tfrac{k}{\alpha \eps}$. However, by the non-private uniformity testing lower bound, $I(X; M_m) = \Omega(1)$ requires $m = \Omega\left(\tfrac{\sqrt{k}}{\alpha^2}\right)$. This means we have some constant $C'$ such that
        \begin{equation*}
        	m > C'\left(\frac{\sqrt{k}}{\alpha^2}\right)^{1/3}\left(\frac{k}{\alpha \eps}\right)^{2/3} = \Omega\left(\tfrac{k^{5/6}}{\alpha^{4/3}\eps^{2/3}}\right)
        	\stepcounter{equation}\tag{\theequation}\label{eq:alt_lb}
        \end{equation*}
        which suffices for our overall lower bound.
        
        All that remains is the second case, $\tfrac{\alpha(e^\eps-1)(t-1)}{k} < 1$. Then since $e^x \leq 1 + 2x$ for $x \in [0,1]$, $e^{\tfrac{\alpha(e^\eps-1)(t-1)}{k}} \leq 1 + 2\tfrac{\alpha(e^\eps-1)(t-1)}{k}$. Again using $\eps = O(1)$, there exists constant $C_1$ such that $\left[1 + \tfrac{\alpha(e^\eps - 1)}{k}\right]^{t-1} \leq e^{\tfrac{\alpha(e^\eps-1)(t-1)}{k}} \leq 1 + C_1\tfrac{\alpha\eps(t-1)}{k}$. Thus we return to Equation~\ref{eq:8} and get
        $$\P{}{M_{t-1} = m^*} \leq \P{}{M_{t-1} = m^* \mid N_{j,t-1}^{\alpha} = 0} \cdot \left(1 + C_1\tfrac{\alpha\eps(t-1)}{k}\right).$$
        If we repeat this process using the other direction of group privacy, we get
        $$\P{}{M_{t-1} = m^*} \geq \P{}{M_{t-1} = m^* \mid N_{j,t-1}^{\alpha} = 0}\left[1 + \frac{\alpha(e^{-\eps}-1)}{k}\right]^{t-1}.$$
        $k \geq 2$, $\eps > 0$, and $\alpha \leq 1$, so $\tfrac{\alpha(e^{-\eps}-1)}{k} \in (-1,0)$. Thus $\left[1 + \tfrac{\alpha(e^{-\eps}-1)}{k}\right]^{t-1} \geq 1 + \tfrac{\alpha(e^{-\eps}-1)(t-1)}{k}$. By $\eps = O(1)$, we get a constant $C_2$ such that $\left[1 + \tfrac{\alpha(e^{-\eps}-1)}{k}\right]^{t-1} \geq 1 - C_2 \tfrac{\alpha \eps (t-1)}{k}$. Tracing back, $$\P{}{M_{t-1} = m^*} \geq \P{}{M_{t-1} = m^* \mid N_{j,t-1}^{\alpha} = 0} \cdot \left(1 - C_2\tfrac{\alpha\eps(t-1)}{k}\right).$$
        
        Returning to the beginning of our proof, we can repeat the argument for the numerator of $\tfrac{\P{}{M_{t-1} = m^* \mid Y_j = 1}}{\P{}{M_{t-1} = m^*}}$:
        \begin{align*}
            \P{}{M_{t-1} = m^* \mid Y_j = 1} =&\; \sum_{q=0}^{t-1} \P{}{M_{t-1} = m^* \mid N_{j,t-1}^{\alpha} = q, Y_j = 1} \cdot \P{}{N_{j,t-1}^{\alpha} = q \mid Y_j = 1} \\
            =&\; \sum_{q=0}^{t-1} \P{}{M_{t-1} = m^* \mid N_{j,t-1}^{\alpha} = q, Y_j = 1} \cdot \P{}{N_{j,t-1}^{\alpha} = q}
        \end{align*}
        since $N_{j,t}^{\alpha}$ and $Y_j$ are independent. Fixing a $q$, we rewrite $\P{}{M_{t-1} = m^* \mid N_{j,t-1}^{\alpha} = q, Y_j = 1}$
        \begin{align*}
             =&\; \sum_{s \in \S_{j,q,t-1}^{\alpha}} \P{}{M_{t-1} = m^* \mid S_{j, \leq t-1} = s, Y_j = 1} \cdot \P{}{S_{j, \leq t-1} = s \mid N_{j,t-1}^{\alpha} = q} \\
            =&\; \sum_{s \in \S_{j,q,t-1}^{\alpha}} \frac{1}{\binom{t-1}{q}k^{t-1-q}} \cdot \P{}{M_{t-1} = m^* \mid S_{j, \leq t-1} = s, Y_j = 1}
            \stepcounter{equation}\tag{\theequation}\label{eq:6}
        \end{align*}
        where the first equality uses the independence of $M_{t-1}$ and $N_{j,t-1}^{\alpha}$ conditioned on $S_{j,t-1}$ as well as the independence of $S_{j, \leq t-1}$ and $Y_j$, and the second equality uses the same counting argument as in the denominator case. Next, $\eps$-pan-privacy gives 
        $$\P{}{M_{t-1} = m^* \mid S_{j, \leq t-1} = s, Y_j = 1} \leq e^{q\eps} \P{}{M_{t-1} = m^* \mid S_{j, \leq t-1} = s', Y_j = 1}$$
        and so
        \begin{align*}
        		(\ref{eq:6}) =&\; \sum_{s \in \S_{j,q,t-1}^{\alpha}} \frac{1}{k^q} \sum_{s' \in T_{s,q}} \frac{1}{\binom{t-1}{q}k^{t-1-q}} \cdot \P{}{M_{t-1} = m^* \mid S_{j,t-1} = s, Y_j = 1} \\
        		\leq&\; \sum_{s \in \S_{j,q,t-1}^{\alpha}} \sum_{s' \in T_{s,q}} \frac{e^{q \eps}}{\binom{t-1}{q}k^{t-1}} \P{}{M_{t-1} = m^* \mid S_{j,t-1} = s', Y_j = 1} \\
        		=&\; \sum_{s' \in \S_{j,q,t-1}^{\alpha}} \frac{e^{q \eps}}{k^{t-1}} \cdot \P{}{M_{t-1} = m^* \mid S_{j,t,-1} = s', Y_j = 1} \\
        		=&\; \sum_{s' \in \S_{j,0,t-1}^{\alpha}} e^{q \eps} \cdot \P{}{M_{t-1} = m^* \mid S_{j, \leq t-1} = s', Y_j = 1} \cdot \P{}{S_{j, \leq t-1} = s' \mid N_{j,t-1}^{\alpha} = 0, Y_j = 1} \\
        		=&\; e^{q \eps} \P{}{M_{t-1} = m^* \mid N_{j, t-1}^{\alpha} = 0}
        \end{align*}
        where the last equality uses the independence of $S_{j, \leq t-1}$ and $Y_j$ conditioned on $N_{j,t-1}^{\alpha} = 0$ and the independence of $M_{t-1}$ and $Y_j$ and $N_{j,t-1}^{\alpha}$ conditioned on $S_{j, \leq t-1}$. In turn we get
        $$\P{}{M_{t-1} = m^* \mid Y_j = 1} \leq \P{}{M_{t-1} = m^* \mid N_{j,t-1}^{\alpha} = 0} \sum_{q=0}^{t-1} e^{q\eps} \P{}{N_{j,t-1}^{\alpha} = q}$$
        which is the same quantity as in Equation~\ref{eq:8}. The same analysis thus gives
        $$\P{}{M_{t-1} = m^* \mid Y_j = 1} \leq \P{}{M_{t-1} = m^* \mid N_{j,t-1}^{\alpha}  = 0} \cdot \left(1 + C_1 \frac{\alpha\eps(t-1)}{k}\right)$$
        as in the denominator case, and
        $$\P{}{M_{t-1} = m^* \mid Y_j = 1} \geq \P{}{M_{t-1} = m^* \mid N_{j,t-1}^{\alpha} = 0} \cdot \left(1 - C_2\frac{\alpha\eps(t-1)}{k}\right).$$
        
        Summing up, both $\P{}{M_{t-1} = m^*}$ and $\P{}{M_{t-1} = m^* \mid Y_j = 1}$ lie in the interval
        $$\left[\P{}{M_{t-1} = m^* \mid N_{j,t-1}^{\alpha} = 0} \cdot \left(1 - C_2\frac{\alpha\eps(t-1)}{k}\right), \P{}{M_{t-1} = m^* \mid N_{j,t-1}^{\alpha} = 0} \cdot \left(1 + C_1\frac{\alpha\eps(t-1)}{k}\right)\right].$$
        Thus
        \begin{align*}
        	\frac{\P{}{M_{t-1} = m^* \mid Y_j = 1}}{\P{}{M_{t-1} = m^*}} \leq&\; \frac{1 + C_1\frac{\alpha\eps(t-1)}{k}}{1 - C_2\frac{\alpha\eps(t-1)}{k}} \\
        	=&\; 1 + \frac{C_1 + C_2}{1 - C_2\frac{\alpha\eps(t-1)}{k}} \cdot \frac{\alpha\eps(t-1)}{k} \\
        	=&\; 1 + O\left(\frac{\alpha\eps t}{k}\right)
        	\end{align*}
        	where the last equality uses $\tfrac{\alpha \eps (t-1)}{k} < \tfrac{1}{2C_2}$ (otherwise, we get $m = \Omega\left(\tfrac{k}{\alpha \eps}\right)$ and can use the argument given in Equation~\ref{eq:alt_lb}). Similarly,
       \begin{align*}
       \frac{\P{}{M_{t-1} = m^* \mid Y_j = 1}}{\P{}{M_{t-1} = m^*}} \geq&\; \frac{1 - C_2\frac{\alpha \eps (t-1)}{k}}{1 + C_1\frac{\alpha \eps (t-1)}{k}} \\
       =&\; 1 - \frac{C_1 + C_2}{1 + C_1\frac{\alpha \eps (t-1)}{k}} \cdot \frac{\alpha \eps (t-1)}{k} \\
       =&\; 1 - O\left(\frac{\alpha\eps t}{k}\right)
       \end{align*}
       and the claim follows.
    \end{proof}
    
    Lemma~\ref{lem:near_one} gives $A \leq \tfrac{\alpha^2\eps^2t^2}{k}$, so $\tfrac{\alpha^2A}{k} \leq \tfrac{\alpha^4\eps^2t^2}{k^2}$. Returning to Equation~\ref{eq:branch} and using $t \leq m$, $I(X; V_t \mid M_{t-1}, J_t) = O\left(\frac{\alpha^4\eps^2m^2}{k^2}\right)$. Then we trace back to Equation~\ref{eq:4} and get $I(X;M_m) = O\left(\frac{\alpha^4\eps^2m^3}{k^2}\right)$. Finally, a uniformity tester requires $I(X; M_m) = \Omega(1)$, so $m = \Omega\left(\frac{k^{2/3}}{\alpha^{4/3}\eps^{2/3}}\right)$. 
\end{proof}
\subsection{Locally Private Lower Bound}
\label{sec:lp_lb}
We now move to the locally private lower bound. We state our result for $\eps$-locally private algorithms, but this is without loss of generality by the work of~\citet{BNS18} and~\citet{CSUZZ19}, which demonstrates an equivalence between $(\eps,\delta)$- and $(\eps,0)$-local privacy for reasonable parameter ranges.

At a high level, the main difference the pan-private and sequentially interactive lower bounds is that the locally private algorithm does not see any sample $S_t$. Instead, the algorithm sees a randomizer output based on $S_t$. We can therefore use past work quantifying the information loss between a randomizer's input and output~\cite{DJW13} to bound information learned more tightly than under pan-privacy. This partially explains, for example, the locally private lower bound's different dependence on $\eps$. Replacing the memory upper bound used by~\citet{DGKR19} with the local privacy restriction also requires a different argument than in the pan-private case.

\begin{theorem}
\label{thm:silb}
    For $\eps = O(1)$, any sequentially interactive $\eps$-locally private uniformity tester requires $m = \Omega\left(\tfrac{k}{\alpha^2\eps^2}\right)$ samples.
\end{theorem}
\begin{proof}
    Let $M_t$ be the random variable for the message sent by user $t$ with sample $S_t$, and let $M_{1:t}$ be the concatenation of messages sent through time $t$. We start by distinguishing our approach for this lower bound from its pan-private analogue. Recall that in the pan-private lower bound we expressed the mutual information between the distribution parameter $X$ and the internal state after $m$ samples $M_m$ as $I(X; M_m) = \sum_{t=1}^m I(X; S_t \mid M_{t-1})$. Here, we want to control the mutual information between $X$ and the transcript through $m$ samples, $I(X; M_{1:m})$. A key difference in the local setting is that the algorithm does not see any sample $S_t$. Instead, the algorithm sees a randomizer output based on $S_t$. We should therefore expect some information loss between the sample and its randomizer output. We formalize this using existing local privacy work (Lemma~\ref{lem:djw}) and get $I(X; M_{1:m}) < \sum_{t=1}^m O(\eps^2) \cdot I(X; S_t \mid M_{1:t-1})$. This partially explains the locally private lower bound's different dependence on $\eps$.
    
    More formally, by the chain rule for mutual information, $I(X; M_{1:m}) = \sum_{t=1}^m I(X; M_t \mid M_{1:t-1})$. Choose one term $I(X;M_t \mid M_{1:t-1})$ and fix a value $m$ for $M_{1:t-1}$. We can rewrite $I(X;M_t \mid M_{1:t-1} = m)$ as
    \begin{align*}
    		=&\; \E{X \mid M_{1:t-1} = m}{\kl{M_t \mid X, M_{1:t-1} = m}{M_t \mid M_{1:t-1} = m}} \\
    		=&\; \P{}{X = 0 \mid M_{1:t-1} = m} \kl{M_t \mid X = 0, M_{1:t-1} = m}{M_t \mid M_{1:t-1} = m} \\
    		&\;+ \P{}{X = 1 \mid M_{1:t-1} = m} \kl{M_t \mid X = 1, M_{1:t-1} = m}{M_t \mid M_{1:t-1} = m}.
    		\stepcounter{equation}\tag{\theequation}\label{eq:si_0}
    \end{align*}
    $M_{1:m}$ is generated by a sequentially interactive $\eps$-locally private protocol. We can therefore use the following result from Duchi et al.~\cite{DJW13}.
    
\begin{lemma}[Theorem 1~\cite{DJW13}]
\label{lem:djw}
	Let $Q$ be the output distribution for an $\eps$-local randomizer in a sequentially interactive protocol. For any two input distributions $P_1$ and $P_2$, the induced output distributions $Q_1$ and $Q_2$ have
	$$\kl{Q_1}{Q_2} + \kl{Q_2}{Q_1} \leq 4(e^\eps-1)^2\tv{P_1}{P_2}^2.$$
\end{lemma}

Here, we let $P_1$ be the distribution for $S_t \mid M_{1:t-1} = m$, $P_2$ for $S_t \mid X = 0, M_{1:t-1} = m$, and $P_3$ for $S_t \mid X = 1, M_{1:t-1} = m$. $Q_1$ is then the distribution for $M_t \mid M_{1:t-1} = m$, $Q_2$ for $M_t \mid X = 0, M_{1:t-1} = m$, and $Q_3$ for $M_t \mid X = 1, M_{1:t-1} = m$. Lemma~\ref{lem:djw} then gives
	\begin{align*}
		(\ref{eq:si_0}) \leq&\; 4(e^\eps-1)^2\left[\P{}{X = 0 \mid M_{1:t-1} = m}\tv{P_1}{P_2}^2 + \P{}{X = 1 \mid M_{1:t-1} = m}\tv{P_1}{P_3}^2\right] \\
		\leq&\; 2(e^\eps-1)^2 [\P{}{X = 0 \mid M_{1:t-1} = m} \kl{P_1}{P_2} + \P{}{X=1 \mid M_{1:t-1} = m} \kl{P_1}{P_3}] \\
		=&\; 2(e^\eps-1)^2 I(X; S_t \mid M_{1:t-1} = m) 
	\end{align*}
where the second inequality uses Pinsker's inequality (Lemma~\ref{lem:pinsker} in the Appendix). Now we can quantify the loss in information between the sample $S_t$ and the private message $M_t$:
	\begin{align*}
		I(X; M_{1:m}) =&\; \sum_{t=1}^m I(X; M_{t} \mid M_{1:t-1}) \\
		\leq&\; \sum_{t=1}^m 2(e^\eps - 1)^2 I(X; S_t \mid M_{1:t-1}) \\
		 \leq&\; \sum_{t=1}^m 2(e^\eps-1)^2 I(X; V_t \mid M_{1:t-1}, J_t)
		\stepcounter{equation}\tag{\theequation}\label{eq:si_beg}
	\end{align*}
    and, by the same reasoning as in the proof of Theorem~\ref{thm:pplb},
    \begin{align*}
    I(X; V_t \mid M_{1:t-1}, J_t) =&\ O\left(\frac{\alpha^2}{k}\E{M_{1:t-1}}{\sum_{j=1}^k \E{}{Y_j \mid M_{1:t-1}}^2}\right) \\
	=&\ O\left(\frac{\alpha^2}{k}\sum_{j=1}^k \E{M_{1:t-1}}{\E{}{Y_j \mid M_{1:t-1}}^2}\right)
    \stepcounter{equation}\tag{\theequation}\label{eq:si_sec}.
    \end{align*}
    Next, we choose a term $j$ of the sum in Equation~\ref{eq:si_sec} and upper bound it. We first rewrite it to incorporate $Y_{-j} = (Y_1, Y_2, \ldots, Y_{j-1}, Y_{j+1}, \ldots, Y_k)$, i.e. the random variable for all $Y_{j'}$ where $j' \neq j$. Incorporating $Y_{-j}$ will be useful for controlling  dependencies between messages and the $Y_j$ later in the argument. Let $U_j$ denote the set of possible realizations for $Y_{-j}$. Then we expand
$$\E{M_{1:t-1}}{\E{}{Y_j \mid M_{1:t-1}}^2} = \E{M_{1:t-1}}{\left( \sum_{u \in U_j} \P{}{Y_{-j} = u \mid M_{1:t-1}}\E{}{Y_j \mid M_{1:t-1}, Y_{-j} = u}\right)^2}$$
and use Cauchy-Schwarz to upper bound the squared sum by
\begin{align*}
	&\ \left(\sum_{u \in U_j} \P{}{Y_{-j} = u \mid M_{1:t-1}} \E{}{Y_j \mid M_{1:t-1}, Y_{-j} = u}^2\right) \cdot \left(\sum_{u \in U_j} \P{}{Y_{-j} = u \mid M_{1:t-1}}\right) \\
	=&\ \E{Y_{-j}}{\E{}{Y_j \mid M_{1:t-1}, Y_{-j}}^2} \cdot 1.
\end{align*}
Returning to Equation~\ref{eq:si_sec} gives
$$ I(X; V_t \mid M_{1:t-1}, J_t) = O\left(\frac{\alpha^2}{k}  \sum_{j=1}^k \E{M_{1:t-1}, Y_{-j}}{\E{}{Y_j \mid M_{1:t-1}, Y_{-j}}^2}\right)$$    
and in turn we rewrite the RHS inside $O\left(\cdot\right)$ as 
    \begin{equation}
    \frac{\alpha^2}{k}\sum_{i=1}^{t-1} \sum_{j=1}^k\left(\E{M_{1:i}, Y_{-j}}{\E{}{Y_j \mid M_{1:i}, Y_{-j}}^2} - \E{M_{1:i-1}, Y_{-j}}{\E{}{Y_j \mid M_{1:i-1}, Y_{-j}}^2}\right). \stepcounter{equation}\tag{\theequation}\label{eq:3}
    \end{equation}
    We now fix some $i$ and want to upper bound
    $$\sum_{j=1}^k\left(\E{M_{1:i}, Y_{-j}}{\E{}{Y_j \mid M_{1:i}, Y_{-j}}^2} - \E{M_{1:i-1}, Y_{-j}}{\E{}{Y_j \mid M_{1:i-1}, Y_{-j}}^2}\right).$$
    Choose one term $j$ and define $\gamma_j = \P{}{Y_j = 1 \mid M_{1:i}, Y_{-j}}$. Then we get 
    \begin{align*}
        \E{M_{1:i}, Y_{-j}}{\E{}{Y_j \mid M_{1:i}, Y_{-j}}^2} =&\; \E{M_{1:i}, Y_{-j}}{\left(\gamma_j - (1 - \gamma_j)\right)^2} \\
        =&\; \E{M_{1:i}, Y_{-j}}{4\gamma_j^2 - 4\gamma_j + 1} \\
        =&\; 4\E{M_{1:i}, Y_{-j}}{\gamma_j^2} - 4\E{M_{1:i}, Y_{-j}}{\gamma_j} + 1 \\
        =&\; 4\E{M_{1:i}, Y_{-j}}{\gamma_j^2} - 1
    \end{align*}
    where the last equality uses
    $$4\E{M_{1:i}, Y_{-j}}{\gamma_j} = 4\E{M_{1:i}, Y_{-j}}{\P{}{Y_j = 1 \mid M_{1:i}, Y_{-j}}} = 4\P{}{Y_j = 1} = 2.$$By similar reasoning, if we define $\eta_j = \P{}{Y_j = 1 \mid M_{1:i-1}, Y_{-j}}$ then we get
    $$\E{M_{1:i-1}, Y_{-j}}{\E{}{Y_j \mid M_{1:i-1}, Y_{-j}}^2} = 4\E{M_{1:i-1}, Y_{-j}}{\eta_j^2} - 1.$$
    Tracing back, our goal is now to upper bound
    \begin{align*}
        &\ \E{M_{1:i}, Y_{-j}}{\E{}{Y_j \mid M_{1:i}, Y_{-j}}^2} - \E{M_{1:t-1}, Y_{-j}}{\E{}{Y_j \mid M_{1:i-1}, Y_{-j}}^2} \\
        =&\ 4\left(\E{M_{1:i}, Y_{-j}}{\gamma_j^2} - \E{M_{1:i-1}, Y_{-j}}{\eta_j^2}\right)
        \stepcounter{equation}\tag{\theequation}\label{eq:gammas}.
    \end{align*}
    Our analysis will be easier if we restrict the message space for $M_1, \ldots, M_i$ to be binary. We do so by a result from Bassily and Smith~\cite{BS15}. This again relies on the local privacy of the protocol.
    
    \begin{lemma}[Theorem 4.1~\cite{BS15}]
    \label{lem:bs}
        Given a sequentially interactive $\eps$-locally private protocol with expected number of randomizer calls $T$, there exists an equivalent sequentially interactive $\eps$-locally private protocol with expected sample complexity $e^\eps T$ where each user sends a single bit from a single randomizer call.
    \end{lemma}
    
    The cost of this transformation is an $e^\eps$ blowup in expected sample complexity and an additional $O(n\log(\log(n)))$ bits of public randomness. First, since we assumed $\eps = O(1)$, by Markov's inequality we can trade an arbitrarily small constant $c$ decrease in overall success probability for a constant ($O(e^\eps/c) = O(1)$) blowup in sample complexity. Combined with our assumption of arbitrary access to public randomness for locally private protocols, it is without loss of generality to assume all of our $M_1, \ldots, M_i$ are binary.\footnote{Note that this step relies on the fact that, in sequentially interactive protocols, the number of randomizer calls is the same as the sample complexity. For fully interactive protocols, the number of randomizer calls may arbitrarily exceed the sample complexity. However, using the transformation given by Joseph et al.~\cite{JMNR19}, our argument also extends to any $O(1)$-compositional fully interactive protocol.}
    
    Returning to Equation~\ref{eq:gammas}, fix $M_{1:i-1}$ and $Y_{-j}$ below. Then
    \begin{align*}
        \E{M_{1:i}, Y_{-j}}{\gamma_j^2} =&\; \P{}{M_i=1} \cdot \P{}{Y_j = 1 \mid M_i = 1}^2 + \P{}{M_i = 0} \cdot \P{}{Y_j = 1 \mid M_i = 0}^2 \\
        =&\; \frac{\left[\P{}{M_i = 1 \mid Y_j = 1} \cdot \P{}{Y_j = 1}\right]^2}{\P{}{M_i = 1}} + \frac{\left[\P{}{M_i = 0 \mid Y_j = 1} \cdot \P{}{Y_j = 1}\right]^2}{\P{}{M_i = 0}} \\
        =&\; \eta_j^2\left[\frac{\P{}{M_i = 1 \mid Y_j = 1} ^2}{\P{}{M_i = 1}} + \frac{\P{}{M_i = 0 \mid Y_j = 1}^2}{\P{}{M_i = 0}}\right]
    \end{align*}
    where the second equality uses Bayes' rule. Now, using $-2x + 2y - 2(1 - x) + 2(1-y) = 0$ with $x = \P{}{M_i = 1 \mid Y_j = 1}$ and $y = \P{}{M_i = 1}$, we get
    $$-2\P{}{M_i = 1 \mid Y_j = 1} + 2\P{}{M_i = 1} - 2\P{}{M_i = 0 \mid Y_j = 1} + 2\P{}{M_j = 0} = 0.$$
    We can now add 0 inside the bracketed term to get
    $$\eta_j^2\left[\frac{\P{}{M_i = 1 \mid Y_j = 1} ^2}{\P{}{M_i = 1}} + \frac{\P{}{M_i = 0 \mid Y_j = 1}^2}{\P{}{M_i = 0}}\right] = \eta_j^2\left[A + B\right]$$
    where
    \begin{align*}
        A =&\; \frac{\P{}{M_i = 1 \mid Y_j = 1}^2 - 2\P{}{M_i = 1 \mid Y_j = 1}\P{}{M_i = 1} + 2\P{}{M_i = 1}^2}{\P{}{M_i = 1}} \\
        =&\; \frac{\left(\P{}{M_i = 1 \mid Y_j = 1} - \P{}{M_i = 1}\right)^2}{\P{}{M_i = 1}} + \P{}{M_i=1}
    \end{align*}
    and
    \begin{align*}
        B =&\; \frac{\P{}{M_i = 0 \mid Y_j = 1}^2 - 2\P{}{M_i = 0 \mid Y_j = 1}\P{}{M_i = 0} + 2\P{}{M_i = 0}^2}{\P{}{M_i=0}} \\
        =&\; \frac{\left(\P{}{M_i = 0 \mid Y_j = 1} - \P{}{M_i = 0}\right)^2}{\P{}{M_i = 0}} + \P{}{M_i=0}.
    \end{align*}
    Thus we may rewrite $\eta_j^2[A + B]$ as
    \begin{equation}
    \label{eq:2}
        \eta_j^2\left[1 + \frac{\left(\P{}{M_i = 1 \mid Y_j = 1} - \P{}{M_i = 1}\right)^2}{\P{}{M_i = 1}} + \frac{\left(\P{}{M_i = 0 \mid Y_j = 1} - \P{}{M_i = 0}\right)^2}{\P{}{M_i = 0}}\right].
    \end{equation}
    For neatness, let $C = \P{}{M_i = 1 \mid Y_j = 1, J_i = j}$ and $D = \P{}{M_i = 1 \mid Y_j = -1, J_i = j}$. Recall that $J_i$ denotes which of $k$ bin pairs is chosen in step $i$. Then
    \begin{align*}
        \P{}{M_i = 1 \mid Y_j = 1} =&\; \P{}{M_i = 1 \mid Y_j = 1, J_i \neq j} \cdot \P{}{J_i \neq j \mid Y_j = 1} \\
        &+\; \P{}{M_i = 1 \mid Y_j = 1, J_i = j} \cdot \P{}{J_i = j \mid Y_j = 1} \\
        =&\; \frac{k-1}{k} \cdot \P{}{M_i = 1 \mid Y_j = 1, J_i \neq j} + \frac{C}{k}
    \end{align*}
    since $J_i$ is independent of $Y_j$ and $\P{}{J_i = j} = \tfrac{1}{k}$. Similarly, 
    \begin{align*}
        \P{}{M_i = 1} =&\; \P{}{M_i = 1 \mid J_i \neq j} \cdot \P{}{J_i \neq j} + \P{}{M_i = 1 \mid J_i = j} \cdot \P{}{J_i = j} \\
        =&\; \frac{k-1}{k} \cdot \P{}{M_i = 1 \mid Y_j = 1, J_i \neq j} \\
        &+\; \frac{1}{k} \cdot (\P{}{M_i = 1 \mid J_i = j, Y_j = 1} \cdot \P{}{Y_j = 1} + \P{}{M_i = 1 \mid J_i = j, Y_j = -1} \cdot \P{}{Y_j = - 1}) \\
        =&\; \frac{k-1}{k} \cdot \P{}{M_i = 1 \mid Y_j = 1, J_i \neq j} + \frac{1}{k}\left(\eta_j C + (1-\eta_j)D\right)
    \end{align*}
    where the second equality uses the fact that conditioned on $M_{1:i-1}$, $Y_{-j}$, and $J_i \neq j$, $M_i$ is independent of $Y_j$. This is because conditioning on $M_{1:i-1}$ alone may introduce dependence among the different $Y_{j'}$, in which case $M_i$ may not be independent of $Y_j$ even conditioned on $J_i \neq j$. However, additionally conditioning on $Y_{-j}$ as we do here breaks this dependence between $M_i$ and $Y_j$ conditioned on $J_i \neq j$, as a sample from any bin pair other than $j$ now no longer adds information about $Y_j$. This is why we introduced $Y_{-j}$ earlier.
    
    We substitute these expressions for $\P{}{M_i = 1 \mid Y_j = 1}$ and $\P{}{M_i = 1}$ and get
    $$\left(\P{}{M_i = 1 \mid Y_j = 1} - \P{}{M_i = 1}\right)^2 = \left[\frac{(1 - \eta_j)(C-D)}{k}\right]^2 = (\P{}{M_i = 0 \mid Y_j = 1} - \P{}{M_i = 0})^2$$
    where the last equality follows from $\P{}{M_i = 0 \mid Y_j = 1} = 1 - \P{}{M_i = 1 \mid Y_j = 1}$ and $\P{}{M_i = 1} = 1 - \P{}{M_i = 0}$. Returning to Equation~\ref{eq:2}, we have
    \begin{align*}
        \eta_j^2\left[A + B\right] =&\; \eta_j^2\left[1 + \left(\frac{(1 - \eta_i)(C-D)}{k}\right)^2\left(\frac{1}{\P{}{M_i = 1}} + \frac{1}{\P{}{M_i = 0}}\right)\right] \\
        =&\; \eta_j^2\left[1 + \left(\frac{(1 - \eta_i)(C-D)}{k}\right)^2 \cdot \frac{1}{\P{}{M_i = 1} \P{}{M_i = 0}}\right]
        \stepcounter{equation}\tag{\theequation}\label{eq:A_B}
    \end{align*}
    since $\P{}{M_i = 1} + \P{}{M_i = 0} = 1$. We now analyze $\tfrac{|C - D|}{\P{}{M_i = 1}}$. It will be useful to recall the sampling thought experiment used in the proof of Lemma~\ref{lem:near_one}: at each time $t$, we first uniformly sample bin pair $J_t \sim_U [k]$ and then sample the bin from a mixture: having sampled bin pair $j$, with probability $1-\alpha$ we take a uniform random draw from $\{2j-1, 2j\}$. With the remaining probability $\alpha$, if $Y_j = 1$ then we sample $2j-1$, and if $Y_j = -1$ then we sample $2j$. Finally, we define $E_{j,t}^{\alpha} = 1$ if $J_t = j$ and we sample from the $\alpha$ mixture component and $E_{j,t}^{\alpha} = 0$ otherwise.
    
    Under this equivalent sampling method, we can rewrite
    \begin{align*}
    		C =&\; \P{}{M_i = 1 \mid Y_j = 1, J_i = j} \\
    		=&\; \P{}{M_i = 1 \mid E_{j,i}^{\alpha} = 1, Y_j = 1, J_i = j} \P{}{E_{j,i}^{\alpha} = 1 \mid Y_j = 1, J_i = j} \\
    		&\;+ \P{}{M_i = 1 \mid E_{j,i}^{\alpha} = 0, Y_j = 1, J_i = j} \P{}{E_{j,i}^{\alpha} = 0 \mid Y_j = 1, J_i = j} \\
    		=&\; \alpha \P{}{M_i = 1 \mid Y_j = 1, E_{j,i}^{\alpha} = 1}
    		+ (1-\alpha)\P{}{M_i = 1 \mid E_{j,i}^{\alpha} = 0, J_i = j}
    \end{align*}
    where the last equality uses the fact that $M_i$ is independent of $J_i$ conditioned on $E_{j,i}^{\alpha} = 1$ and $M_i$ is independent of $Y_j$ conditioned on $E_{j,i}^{\alpha} = 0$, $M_{1:i-1}$, and $Y_{-j}$. Similarly
    $$D = \alpha \P{}{M_i = 1 \mid Y_j = -1, E_{j,i}^{\alpha} = 1}
    		+ (1-\alpha)\P{}{M_i = 1 \mid E_{j,i}^{\alpha} = 0, J_i = j}.$$
    	Thus we can rewrite
    \begin{align*}
    		\frac{|C - D|}{\P{}{M_i = 1}} =&\; \frac{|\alpha(\P{}{M_i = 1 \mid Y_j = 1, E_{j,i}^{\alpha} = 1} - \P{}{M_i = 1 \mid Y_j = -1, E_{j,i}^{\alpha} = 1})|}{\P{}{M_i = 1}} \\
    		\leq&\; \frac{|\alpha(e^{\eps} - e^{-\eps})\P{}{M_i = 1}|}{\P{}{M_i = 1}} \\
    		=&\; O(\alpha \eps)
    \end{align*}
    where the inequality uses the $\eps$-local privacy of $M_i$ (recalling that we have been conditioning on $M_{1:i-1}$), and the equality uses $\eps = O(1)$. Similarly, we get 
    \begin{align*}
    		1 - C =&\; \P{}{M_i = 0 \mid Y_j = 1, J_i = j} \\
    		=&\; \alpha \P{}{M_i = 0 \mid Y_j = 1, E_{j,i}^{\alpha} = 1}
    		+ (1-\alpha)\P{}{M_i = 0 \mid E_{j,i}^{\alpha} = 0, J_i = j}
    \end{align*}
    and
    $$1 - D = \alpha \P{}{M_i = 0 \mid Y_j = -1, E_{j,i}^{\alpha} = 1} + (1-\alpha) \P{}{M_i = 0 \mid E_{j,i}^{\alpha} = 0, J_i = j}.$$
    This gives us 
    \begin{align*}
    		\frac{|C-D|}{\P{}{M_i = 0}} =&\; \frac{|(1-C) - (1-D)|}{\P{}{M_i = 0}} \\
    		=&\; \frac{|\alpha(\P{}{M_i = 0 \mid Y_j = 1, E_{j,i}^{\alpha} = 1} - \P{}{M_i = 0 \mid Y_j = -1, E_{j,i}^{\alpha} = 1})|}{\P{}{M_i = 1}} \\
    		\leq&\; \frac{|\alpha(e^\eps - e^{-\eps})\P{}{M_i = 0}}{\P{}{M_i = 0}} \\
    		=&\; O(\alpha \eps)
    \end{align*}
    as well. Thus by Equation~\ref{eq:A_B}, $\eta_j^2[A + B] = \eta_j^2 + O\left(\frac{\eta_j^2(1-\eta_i)^2\alpha^2\eps^2}{k^2}\right) = \eta_j^2 + O\left(\tfrac{\alpha^2\eps^2}{k^2}\right)$ because $\eta_j^2(1 - \eta_j)^2 < 1$. Returning to Equation~\ref{eq:gammas}, we can now bound
    $$\E{M_{1:i}, Y_{-j}}{\E{}{Y_j \mid M_{1:i}, Y_{-j}}^2} - \E{M_{1:t-1}, Y_{-j}}{\E{}{Y_j \mid M_{1:i-1}, Y_{-j}}^2} = O\left(\frac{\alpha^2\eps^2}{k^2}\right).$$
    Since this analysis was for an arbitrary $j$, we get
    $$\sum_{j=1}^k\left(\E{M_{1:i}, Y_{-j}}{\E{}{Y_j \mid M_{1:i}, Y_{-j}}^2} - \E{M_{1:t-1}, Y_{-j}}{\E{}{Y_j \mid M_{1:i-1}, Y_{-j}}^2}\right) = O\left(\frac{\alpha^2\eps^2}{k}\right).$$
    We substitute this into Equations~\ref{eq:3} and~\ref{eq:si_sec} and get $I(X; V_t \mid M_{1:t-1}, J_t) = O\left(\tfrac{\alpha^4\eps^2t}{k^2}\right)$. Finally, substituting back into Equation~\ref{eq:si_beg} and using $t \leq m$ and $\eps = O(1)$, $I(X; M_{1:m}) = O\left(\tfrac{\alpha^4\eps^4m^2}{k^2}\right)$. Since the output of a locally private algorithm is a function of the transcript, a uniformity tester with sample complexity $m$ requires $I(X; M_{1:m}) = \Omega(1)$. We therefore get sample complexity $m = \Omega\left(\tfrac{k}{\alpha^2\eps^2}\right)$.
\end{proof}
\section{Acknowledgments}
\label{sec:ack}
We thank the anonymous reviewers for their helpful comments. These improved the presentation, statements, and proofs of several of our results. We also thank Jayadev Acharya, Clément Canonne, Yuhan Liu, Ziteng Sun, and Himanshu Tyagi for pointing out an error in the original proof of Theorem~\ref{thm:silb}.

\newpage

\bibliographystyle{plainnat}
\bibliography{references}

\begin{thebibliography}{40}
\providecommand{\natexlab}[1]{#1}
\providecommand{\url}[1]{\texttt{#1}}
\expandafter\ifx\csname urlstyle\endcsname\relax
  \providecommand{\doi}[1]{doi: #1}\else
  \providecommand{\doi}{doi: \begingroup \urlstyle{rm}\Url}\fi

\bibitem[Abowd(2018)]{A18}
John Abowd.
\newblock Disclosure avoidance for block level data and protection of
  confidentiality in public tabulations.
\newblock
  \url{census.gov/cac/sac/meetings/2018-12/abowd-disclosure-avoidance.pdf},
  2018.
\newblock Accessed: 04/25/2020.

\bibitem[Acharya et~al.(2013)Acharya, Jafarpour, Orlitsky, and Suresh]{AJOS13}
Jayadev Acharya, Ashkan Jafarpour, Alon Orlitsky, and Ananda Suresh.
\newblock A competitive test for uniformity of monotone distributions.
\newblock In \emph{International Conference on Artificial Intelligence and
  Statistics (AISTATS)}, 2013.

\bibitem[Acharya et~al.(2015)Acharya, Daskalakis, and Kamath]{ADK15}
Jayadev Acharya, Constantinos Daskalakis, and Gautam Kamath.
\newblock Optimal testing for properties of distributions.
\newblock In \emph{Neural Information Processing Systems (NIPS)}, 2015.

\bibitem[Acharya et~al.(2018)Acharya, Sun, and Zhang]{ASZ18}
Jayadev Acharya, Ziteng Sun, and Huanyu Zhang.
\newblock Differentially private testing of identity and closeness of discrete
  distributions.
\newblock In \emph{Neural Information Processing Systems (NeurIPS)}, 2018.

\bibitem[Acharya et~al.(2019{\natexlab{a}})Acharya, Canonne, Freitag, and
  Tyagi]{ACFT19}
Jayadev Acharya, Cl{\'e}ment Canonne, Cody Freitag, and Himanshu Tyagi.
\newblock Test without trust: Optimal locally private distribution testing.
\newblock In \emph{International Conference on Artificial Intelligence and
  Statistics (AISTATS)}, 2019{\natexlab{a}}.

\bibitem[Acharya et~al.(2019{\natexlab{b}})Acharya, Canonne, Han, Sun, and
  Tyagi]{ACHST19}
Jayadev Acharya, Cl{\'e}ment~L Canonne, Yanjun Han, Ziteng Sun, and Himanshu
  Tyagi.
\newblock Domain compression and its application to randomness-optimal
  distributed goodness-of-fit.
\newblock \emph{arXiv preprint arXiv:1907.08743}, 2019{\natexlab{b}}.

\bibitem[Aliakbarpour et~al.(2018)Aliakbarpour, Diakonikolas, and
  Rubinfeld]{ADR18}
Maryam Aliakbarpour, Ilias Diakonikolas, and Ronitt Rubinfeld.
\newblock Differentially private identity and equivalence testing of discrete
  distributions.
\newblock In \emph{International Conference on Machine Learning (ICML)}, 2018.

\bibitem[Apple(2017)]{A17}
Differential Privacy~Team Apple.
\newblock Learning with privacy at scale.
\newblock Technical report, Apple, 2017.

\bibitem[Bassily and Smith(2015)]{BS15}
Raef Bassily and Adam Smith.
\newblock Local, private, efficient protocols for succinct histograms.
\newblock In \emph{Symposium on the Theory of Computing (STOC)}, 2015.

\bibitem[Bittau et~al.(2017)Bittau, Erlingsson, Maniatis, Mironov, Raghunathan,
  Lie, Rudominer, Kode, Tinnes, and Seefeld]{BEMMR+17}
Andrea Bittau, \'{U}lfar Erlingsson, Petros Maniatis, Ilya Mironov, Ananth
  Raghunathan, David Lie, Mitch Rudominer, Ushasree Kode, Julien Tinnes, and
  Bernhard Seefeld.
\newblock Prochlo: Strong privacy for analytics in the crowd.
\newblock In \emph{Symposium on Operating Systems Principles (SOSP)}, 2017.

\bibitem[Blelloch and Golovin(2007)]{BG07}
Guy~E. Blelloch and Daniel Golovin.
\newblock Strongly history-independent hashing with applications.
\newblock In \emph{Foundations of Computer Science (FOCS)}, 2007.

\bibitem[Bun et~al.(2018)Bun, Nelson, and Stemmer]{BNS18}
Mark Bun, Jelani Nelson, and Uri Stemmer.
\newblock Heavy hitters and the structure of local privacy.
\newblock In \emph{Symposium on Principles of Database Systems (PODS)}, 2018.

\bibitem[Cai et~al.(2017)Cai, Daskalakis, and Kamath]{CDK17}
Bryan Cai, Constantinos Daskalakis, and Gautam Kamath.
\newblock Priv'it: private and sample efficient identity testing.
\newblock In \emph{International Conference on Machine Learning (ICML)}, 2017.

\bibitem[Canonne(2015)]{C15}
Cl{\'e}ment~L Canonne.
\newblock A survey on distribution testing: Your data is big. but is it blue?
\newblock In \emph{Electronic Colloquium on Computational Complexity (ECCC)},
  2015.

\bibitem[Canonne(2016)]{C16}
Cl{\'e}ment~L Canonne.
\newblock A short note on poisson tail bounds.
\newblock
  \url{http://www.cs.columbia.edu/~ccanonne/files/misc/2017-poissonconcentration.pdf},
  2016.

\bibitem[Chan et~al.(2014)Chan, Diakonikolas, Valiant, and Valiant]{CDVV14}
Siu-On Chan, Ilias Diakonikolas, Paul Valiant, and Gregory Valiant.
\newblock Optimal algorithms for testing closeness of discrete distributions.
\newblock In \emph{Symposium on Discrete Algorithms (SODA)}, 2014.

\bibitem[Chan et~al.(2011)Chan, Shi, and Song]{CSS11}
T.-H.~Hubert Chan, Elaine Shi, and Dawn Song.
\newblock Private and continual release of statistics.
\newblock \emph{ACM Trans. Inf. Syst. Secur.}, 2011.

\bibitem[Chan et~al.(2012)Chan, Shi, and Song]{CSS12}
TH~Hubert Chan, Elaine Shi, and Dawn Song.
\newblock Optimal lower bound for differentially private multi-party
  aggregation.
\newblock In \emph{European Symposium on Algorithms (ESA)}, 2012.

\bibitem[Cheu et~al.(2019)Cheu, Smith, Ullman, Zeber, and Zhilyaev]{CSUZZ19}
Albert Cheu, Adam Smith, Jonathan Ullman, David Zeber, and Maxim Zhilyaev.
\newblock Distributed differential privacy via shuffling.
\newblock In \emph{Annual International Conference on the Theory and
  Applications of Cryptographic Techniques (CRYPTO)}, 2019.

\bibitem[Diakonikolas et~al.(2019)Diakonikolas, Gouleakis, Kane, and
  Rao]{DGKR19}
Ilias Diakonikolas, Themis Gouleakis, Daniel~M. Kane, and Sankeerth Rao.
\newblock Communication and memory efficient testing of discrete distributions.
\newblock In \emph{Conference on Learning Theory (COLT)}, 2019.

\bibitem[Ding et~al.(2017)Ding, Kulkarni, and Yekhanin]{DKY17}
Bolin Ding, Janardhan Kulkarni, and Sergey Yekhanin.
\newblock Collecting telemetry data privately.
\newblock In \emph{Conference on Neural Information Processing Systems (NIPS)},
  pages 3574--3583, 2017.

\bibitem[Duchi and Rogers(2019)]{DR19}
John Duchi and Ryan Rogers.
\newblock Lower bounds for locally private estimation via communication
  complexity.
\newblock In \emph{Conference on Learning Theory (COLT)}, 2019.

\bibitem[Duchi et~al.(2013)Duchi, Jordan, and Wainwright]{DJW13}
John~C. Duchi, Michael~I. Jordan, and Martin~J. Wainwright.
\newblock Local privacy and statistical minimax rates.
\newblock In \emph{Foundations of Computer Science (FOCS)}. IEEE, 2013.

\bibitem[Dwork et~al.(2006)Dwork, McSherry, Nissim, and Smith]{DMNS06}
Cynthia Dwork, Frank McSherry, Kobbi Nissim, and Adam Smith.
\newblock Calibrating noise to sensitivity in private data analysis.
\newblock In \emph{Theory of Cryptography Conference (TCC)}, 2006.

\bibitem[Dwork et~al.(2010{\natexlab{a}})Dwork, Naor, Pitassi, and
  Rothblum]{DNPR10}
Cynthia Dwork, Moni Naor, Toniann Pitassi, and Guy~N Rothblum.
\newblock Differential privacy under continual observation.
\newblock In \emph{Symposium on the Theory of Computing (STOC)},
  2010{\natexlab{a}}.

\bibitem[Dwork et~al.(2010{\natexlab{b}})Dwork, Naor, Pitassi, Rothblum, and
  Yekhanin]{DNPRY10}
Cynthia Dwork, Moni Naor, Toniann Pitassi, Guy~N Rothblum, and Sergey Yekhanin.
\newblock Pan-private streaming algorithms.
\newblock In \emph{Innovations in Computer Science (ICS)}, 2010{\natexlab{b}}.

\bibitem[Dwork et~al.(2014)Dwork, Roth, et~al.]{DR14}
Cynthia Dwork, Aaron Roth, et~al.
\newblock The algorithmic foundations of differential privacy.
\newblock \emph{Foundations and Trends{\textregistered} in Theoretical Computer
  Science}, 2014.

\bibitem[Goldreich and Ron(2000)]{GR00}
Oded Goldreich and Dana Ron.
\newblock On testing expansion in bounded-degree graphs.
\newblock \emph{Electronic Colloquium on Computational Complexity (ECCC)},
  2000.

\bibitem[Guevara(2019)]{Go19}
Miguel Guevara.
\newblock {Enabling developers and organizations to use differential privacy}.
\newblock
  \url{developers.googleblog.com/2019/09/enabling-developers-and-organizations.html},
  2019.
\newblock Accessed: 09-12-2019.

\bibitem[Joseph et~al.(2019)Joseph, Mao, Neel, and Roth]{JMNR19}
Matthew Joseph, Jieming Mao, Seth Neel, and Aaron Roth.
\newblock The role of interactivity in local differential privacy.
\newblock In \emph{Foundations of Computer Science (FOCS)}. IEEE, 2019.

\bibitem[Joseph et~al.(2020)Joseph, Mao, and Roth]{JMR20}
Matthew Joseph, Jieming Mao, and Aaron Roth.
\newblock Exponential separations in local differential privacy.
\newblock In \emph{Symposium on Discrete Algorithms (SODA)}, 2020.

\bibitem[Kasiviswanathan et~al.(2011)Kasiviswanathan, Lee, Nissim,
  Raskhodnikova, and Smith]{KLNRS11}
Shiva~Prasad Kasiviswanathan, Homin~K Lee, Kobbi Nissim, Sofya Raskhodnikova,
  and Adam Smith.
\newblock What can we learn privately?
\newblock \emph{SIAM Journal on Computing}, 2011.

\bibitem[McGregor et~al.(2010)McGregor, Mironov, Pitassi, Reingold, Talwar, and
  Vadhan]{MMPRTV10}
Andrew McGregor, Ilya Mironov, Toniann Pitassi, Omer Reingold, Kunal Talwar,
  and Salil Vadhan.
\newblock The limits of two-party differential privacy.
\newblock In \emph{Foundations of Computer Science (FOCS)}, 2010.

\bibitem[Messing et~al.(2019)Messing, DeGregorio, Hillenbrand, King, Mahanti,
  Nayak, Persily, State, and Wilkins]{MDHKM+19}
Solomon Messing, Christina DeGregorio, Bennett Hillenbrand, Gary King, Saurav
  Mahanti, Chaya Nayak, Nathaniel Persily, Bogdan State, and Arjun Wilkins.
\newblock Facebook privacy-protected urls light table release.
\newblock
  \url{socialscience.one/files/partnershipone/files/facebook_urls-light_codebook_v2.0.pdf},
  2019.
\newblock Accessed: 09-18-2019.

\bibitem[Micciancio(1997)]{M97}
Daniele Micciancio.
\newblock Oblivious data structures: Applications to cryptography.
\newblock In \emph{Symposium on the Theory of Computing (STOC)}, 1997.

\bibitem[Mir et~al.(2011)Mir, Muthukrishnan, Nikolov, and Wright]{MMNW11}
Darakhshan Mir, Shan Muthukrishnan, Aleksandar Nikolov, and Rebecca~N Wright.
\newblock Pan-private algorithms via statistics on sketches.
\newblock In \emph{Symposium on Principles of Database Systems (PODS)}, 2011.

\bibitem[Murtagh et~al.(2018)Murtagh, Taylor, Kellaris, and Vadhan]{MTKV18}
Jack Murtagh, Kathryn Taylor, George Kellaris, and Salil Vadhan.
\newblock Usable differential privacy: A case study with psi.
\newblock \emph{arXiv preprint arXiv:1809.04103}, 2018.

\bibitem[Naor and Teague(2001)]{NT01}
Moni Naor and Vanessa Teague.
\newblock Anti-persistence: History independent data structures.
\newblock In \emph{Symposium on the Theory of Computing (STOC)}, 2001.

\bibitem[Paninski(2008)]{P08}
Liam Paninski.
\newblock A coincidence-based test for uniformity given very sparsely sampled
  discrete data.
\newblock \emph{IEEE Transactions on Information Theory}, 2008.

\bibitem[Valiant and Valiant(2014)]{VV14}
Gregory Valiant and Paul Valiant.
\newblock An automatic inequality prover and instance optimal identity testing.
\newblock In \emph{Foundations of Computer Science (FOCS)}, 2014.

\end{thebibliography}
\section{Constant Separation in Uniformity Testing}
\label{sec:sep_details}
Recall that Definition~\ref{def:uni} requires success probabilities of at least $2/3$, i.e.
$$\P{}{\text{output uniform} \mid p = U_k} \geq 2/3 \text{ and } \P{}{\text{output uniform} \mid \tv{p}{U_k} \geq \alpha} \leq 1/3.$$
	As long as we achieve constant separation, i.e. have 
$$\P{}{\text{output uniform} \mid p = U_k} \geq c_1 \text{ and } \P{}{\text{output uniform} \mid \tv{p}{U_k} \geq \alpha} \leq c_2$$
for positive $c_1 - c_2 = \Omega(1)$, we can amplify it to a $1/3$ separation by repetition. After sufficiently many repetitions, if $p = U_k$ then the proportion of ``uniform'' answers will concentrate at or above $c_1$, and if $\tv{p}{U_k} \geq \alpha$ it will concentrate at or below $c_2$. By a Chernoff bound, $r = \Omega\left(\tfrac{1}{(c_1 - c_2)^2}\right)$ repetitions suffice to distinguish between these cases. Since this is still a constant number of repetitions, our algorithms will focus on achieving any constant separation.

\section{Uniformity Testing Upper Bound Proofs}
\label{sec:app_b}

\begin{lemma}
	For $m = \Omega\left(\frac{k^{3/4}}{\alpha \eps} + \tfrac{\sqrt{k}}{\alpha^2}\right)$, \ppu~is an $\eps$-pan-private uniformity tester on $m$ samples.
\end{lemma}
\begin{proof}
	\underline{Privacy}: 
	Let $t$ be a time in the stream, let $i$ be a possible internal state for \ppu , and let $o$ be a possible output. Let $p_{\In, s, t}$ be the probability density function for the internal state of \ppu~after the first $t$ elements of stream $s$, and let $p_{\Ou, s, t \mid i}$ be the probability density function for the output given stream $s$ such that the internal state at time $t$ was $i$. Finally, fix neighboring streams $s$ and $s'$. Then to prove that \ppu~is $\eps$-pan-private, it suffices to show that $\tfrac{p_{\In, s, t}(i) \cdot p_{\Ou, s, t \mid i}(o)}{p_{\In, s', t}(i) \cdot p_{\Ou, s', t \mid i}(o)} \leq e^\eps$. 
	
	The final output of \ppu~is a deterministic function of its final internal state (after the second addition of Laplace noise). The final internal state is after $m$ samples, so it is enough to choose arbitrary internal states $i_1$ and $i_2$ and show
	\begin{equation}
	\label{eq:uni_priv}
		\frac{p_{\In, s, t}(i_1) \cdot p_{\In, s, m, t \mid i_1}(i_2)}{p_{\In, s', t}(i_1) \cdot p_{\In, s', m, t \mid i_1}(i_2)} \leq e^\eps.
	\end{equation}
		We first recall a basic fact about differential privacy: if $f$ is a real-valued function with sensitivity $\Delta f$, i.e. a function whose output changes by at most $\Delta$ between neighboring databases, then adding $\lap{\tfrac{\Delta f}{\eps}}$ noise to the output of $f$ is $\eps$-differentially private (see e.g. Theorem 3.4 in the survey of~\citet{DR14}). Here, each bin of $H$ is a 1-sensitive function and each sample alters a single bin. Thus by the first application of $\lap{\tfrac{1}{\eps}}$ noise to each bin we get $\tfrac{p_{\In, s, t}(i_1)}{p_{\In, s', t}(i_1)} \leq e^{\eps}$. Similarly, the second application of $\lap{\tfrac{1}{\eps}}$ noise to each bin implies $\frac{p_{\In, s, m, t \mid i_1}(i_2)}{p_{\In, s', m, t \mid i_1}(i_2)} \leq e^{\eps}$. To get the overall claim, we split into two cases. If $s_{\leq t} = s_{\leq t}'$, then $\tfrac{p_{\In, s, t}(i_1)}{p_{\In, s', t}(i_1)} = 1$. If instead $s_{\leq t} \neq s_{\leq t}'$, then $s_{> t} = s_{> t}'$, so $\frac{p_{\In, s, m, t \mid i_1}(i_2)}{p_{\In, s', m, t \mid i_1}(i_2)} = 1$. Thus Equation~\ref{eq:uni_priv} holds.
	
	\underline{Sample complexity}: To better analyze $Z'$, we decompose it as the sum of a non-private $\chi^2$-statistic $Z$ and a noise term $Y$,
	$$Z = \sum_{i=1}^k \frac{(N_i - m/k)^2 - N_i}{m/k} \text{ and } Y =  \sum_{i=1}^k \frac{[Y_i + Y_i']^2 + 2[Y_i + Y_i'](N_i - m/k) - [Y_i + Y_i']}{m/k}.$$
	where $N_i$ is the true stream count of item $i$ and $Y_i, Y_i' \sim \lap{\tfrac{1}{\eps}}$ are the first and second addition of Laplace noise. This lets us rewrite $Z' = Z + Y$. In the uniform case, we will give a high-probability upper bound for $Z + Y$, and in the non-uniform case we will give a high-probability lower bound. Fortunately,~\citet{ADK15} prove several results about $Z$. We summarize these results in Lemma~\ref{lem:adk}.
	
\begin{lemma}[Lemmas 2 and 3 from~\citet{ADK15}]
\label{lem:adk}
	If $p = U_k$ and $m = \Omega\left(\tfrac{\sqrt{k}}{\alpha^2}\right)$, then $\E{}{Z} \leq \tfrac{\alpha^2 m}{500}$ and $\Var{Z} \leq \tfrac{\alpha^4m^2}{500000}$. If $\tv{p}{U_k} \geq \alpha$, then $\E{}{Z} \geq \tfrac{\alpha^2m}{5}$ and $\Var{Z} \leq \tfrac{\E{}{Z}^2}{100}$.
\end{lemma}

We split into cases depending on $p$. For each case, Lemma~\ref{lem:adk} will control $Z$, and our task will be to control $Y$.
    
    \underline{Case 1}: $p = U_k$. By Lemma~\ref{lem:adk}, $\E{}{Z} \leq \tfrac{\alpha^2m}{500}$ and $\Var{Z} \leq \tfrac{\alpha^4m^2}{500000}$. By Chebyshev's inequality, $\P{}{Z > \left(\tfrac{1}{500} + \tfrac{c}{500\sqrt{2}}\right) \alpha^2m} \leq \tfrac{1}{c^2}$. 
    
    Turning our attention to $Y$, define 
    $$A = \sum_{i=1}^k \frac{[Y_i + Y_i']^2}{m/k} \text{, } B = \sum_{i=1}^k \frac{2[Y_i + Y_i'](N_i - m/k)}{m/k} \text{, and } C = \sum_{i=1}^k \frac{Y_i + Y_i'}{m/k}.$$
    Then we can rewrite $Y = A + B - C$. We control each of $A, B$, and $C$ in turn. First, by the independence of all draws of noise, $\E{}{A} = \tfrac{k^2\E{}{[Y_i + Y_i']^2}}{m} = \tfrac{2k^2\Var{Y_i}}{m} = \tfrac{4k^2}{\eps^2 m}$ because $\Var{\lap{\tfrac{1}{\eps}}} = \tfrac{2}{\eps^2}$. Next,
\begin{align*}
	\Var{A} =&\; \frac{k^3}{m^2}\Var{Y_i^2 + 2Y_iY_i' + Y_i'^2} \\
	=&\; \frac{k^3}{m^2} \left(\E{}{(Y_i^2 + 2Y_iY_i' + Y_i'^2)^2} - \E{}{Y_i^2 + 2Y_iY_i' + Y_i'^2}^2\right) \\
	=&\; \frac{k^3}{m^2} \left(\left[2\E{}{Y_i^4} + 6\E{}{Y_i^2}^2\right] - 4\E{}{Y_i^2}^2\right) \\
	=&\; \frac{2k^3}{m^2}\left(\E{}{Y_i^4} + \E{}{Y_i^2}^2\right) \\
	=&\; \frac{2k^3}{m^2}\left(\frac{12}{\eps^4} + \frac{4}{\eps^4}\right) \\
	=&\; \frac{32k^3}{\eps^4m^2}
\end{align*}

\noindent where we use $\E{}{Y_i^4} = \tfrac{\eps}{2}\int_0^\infty x^4e^{-\eps x}dx = \tfrac{12}{\eps^4}$ by repeated integration by parts. With Chebyshev's inequality, $\P{}{A > \tfrac{4k^2}{\eps^2m} + 6c\tfrac{k^{3/2}}{\eps^2m}} < \tfrac{1}{c^2}$.
    
    To bound $B$, we use $\E{}{B} = 0$ and
    \begin{align*}
        \Var{B} =&\; \frac{4k^2}{m^2} \cdot \Var{\sum_{i=1}^k [Y_i + Y_i']\left(N_i - \frac{m}{k}\right)} \\
        =&\; \frac{4k^2}{m^2} \cdot \E{}{\left(\sum_{i=1}^k [Y_i + Y_i']\left[N_i - \frac{m}{k}\right]\right)^2} \\
        =&\; \frac{4k^2}{m^2} \sum_{i_1, i_2 \in [k]} \E{}{(Y_{i_1} + Y_{i_1}')(Y_{i_2} + Y_{i_2}')} \cdot \E{}{\left(N_{i_1} - \frac{m}{k}\right)\left(N_{i_2} - \frac{m}{k}\right)} \\
        =&\; \frac{4k^2}{m^2} \sum_{i=1}^k \E{}{(Y_i+Y_i')^2} \cdot \E{}{\left(N_i - \frac{m}{k}\right)^2} \\
        =&\; \frac{16k^3}{\eps^2m^2} \left(\E{}{N_1^2} - \frac{2m \E{}{N_1}}{k} + \frac{m^2}{k^2}\right) \\
        =&\; \frac{16k^3}{\eps^2m^2} \left(\Var{N_1} + \E{}{N_1}^2 - \frac{2m^2}{k^2} + \frac{m^2}{k^2}\right) \\
        =&\; \frac{16k^2}{\eps^2 m}
     \end{align*}
where the last two equalities use $N_i \sim \poi{\tfrac{m}{k}}$ and $\Var{\poi{\tfrac{m}{k}}} = \tfrac{m}{k}$. Again applying Chebyshev's inequality gives $\P{}{B > 4c\tfrac{k}{\eps\sqrt{m}}} < \tfrac{1}{c^2}$.

Similarly, $\E{}{C} = 0$, and with $\Var{C} = \tfrac{k^3}{m^2} \cdot \Var{Y_i + Y_i'} = \tfrac{4k^3}{\eps^2m^2}$, $\P{}{C < -2c\frac{k^{3/2}}{\eps m}} \leq \frac{1}{c^2}$.

Combining the above bounds on $Z, A, B,$ and $C$,  with probability at least $1 - \tfrac{4}{c^2}$,
$$Z' \leq \left(\tfrac{1}{500} + \tfrac{c}{500\sqrt{2}}\right) \alpha^2m + \frac{4k^2}{\eps^2m} + 6c\frac{k^{3/2}}{\eps^2m} + 4c\frac{k}{\eps\sqrt{m}} + 2c\frac{k^{3/2}}{\eps m}.$$
Taking $c = 4\sqrt{2}$ and
$$T_U =\tfrac{1}{100}\alpha^2m + 4\tfrac{k^2}{\eps^2m} + 24\sqrt{2}\tfrac{k^{3/2}}{\eps^2m} + 16\sqrt{2} \tfrac{k}{\eps \sqrt{m}} + 8\sqrt{2} \tfrac{k^{3/2}}{\eps m},$$
$\P{}{Z' \leq T_U} \geq 7/8.$

\underline{Case 2}: $\tv{p}{U_k} \geq \alpha$. By Lemma~\ref{lem:adk}, $\E{}{Z} \geq \tfrac{\alpha^2m}{5}$ and $\Var{Z} \leq \tfrac{\E{}{Z}^2}{100}$.  Chebyshev's inequality now gives
$$1 - \frac{1}{c^2} \leq \P{}{Z \geq \E{}{Z} - c\sqrt{\Var{Z}}} \leq \P{}{Z \geq \left(1 - \frac{c}{10}\right)\E{}{Z}} \leq \P{}{Z \geq \left(1 - \frac{c}{10}\right)\frac{\alpha^2m}{5}}$$
where the last inequality requires $c \leq 10$. Returning to the decomposition of $Y$ used in Case 1, $A$ and $C$ are unchanged and we can use our previous expressions for them (with appropriate sign changes for lower bounds). Our last task is to lower bound $B =  \tfrac{2k}{m}\sum_{i=1}^k [Y_i + Y_i'](N_i - m/k)$. For any term $i$, $Y_i$ and $Y_i'$ are symmetric, so
$$\P{}{[Y_i + Y_i'](N_i - m/k) > 0} = \P{}{[Y_i + Y_i'](N_i - m/k) < 0}$$ and $\P{}{B \geq 0} \geq 1/2$.

Summing up, with probability at least $\tfrac{1}{2}-\tfrac{3}{c'^2}$, $$Z' \geq \left(\frac{1}{5} - \frac{c'}{50}\right)\alpha^2m + 4\frac{k^2}{\eps^2m} - 6c'\frac{k^{3/2}}{\eps^2m} - 2c'\frac{k^{3/2}}{\eps m}.$$ Taking $c' = 2\sqrt{3}$ and $T_\alpha =  \tfrac{\alpha^2m}{10} + 4\tfrac{k^2}{\eps^2m} - 12\sqrt{3}\tfrac{k^{3/2}}{\eps^2m} - 4\sqrt{3}\tfrac{k^{3/2}}{\eps m}$, $\P{}{Z' \geq T_\alpha} \geq \frac{1}{4}.$

For $T_\alpha > T_U$, it is enough that $T_\alpha - T_U > 0$.
$$T_\alpha - T_U = \frac{9}{100}\alpha^2m - \left(12\sqrt{3} + 24\sqrt{2}\right)\frac{k^{3/2}}{\eps^2m} - 16\sqrt{2}\frac{k}{\eps\sqrt{m}} - \left(4\sqrt{3} + 8\sqrt{2}\right)\frac{k^{3/2}}{\eps m}.$$
Dropping constants, we need $\alpha^2m = \Omega\left(\tfrac{k^{3/2}}{\eps^2 m} + \tfrac{k}{\eps\sqrt{m}} + \tfrac{k^{3/2}}{\eps m}\right)$.
We can drop the lower-order term $\tfrac{k^{3/2}}{\eps m}$ and get $\alpha^2m = \Omega\left(\tfrac{k^{3/2}}{\eps^2m} + \tfrac{k}{\eps \sqrt{m}}\right)$, i.e. $m = \Omega\left(\tfrac{k^{3/4}}{\alpha \eps} + \tfrac{k^{2/3}}{\alpha^{4/3}\eps^{2/3}}\right)$.

Putting it all together and recalling the assumption from Lemma~\ref{lem:adk}, there exists constant $c$ such that if $m > c\left(\frac{k^{3/4}}{\alpha \eps} + \frac{k^{2/3}}{\alpha^{4/3} \eps^{2/3}} + \tfrac{\sqrt{k}}{\alpha^2}\right)$ then
	$$\P{}{\text{output ``uniform''} \mid \tv{p}{U_k} \geq \alpha} \leq 3/4 \text{ and }\P{}{\text{output ``uniform''} \mid p = U_k} \geq 7/8.$$
	Thus we get a constant $1/8$ separation. By the amplification argument outlined after Definition~\ref{def:uni}, \ppu~is a uniformity tester. Finally, 
	$$\frac{k^{2/3}}{\alpha^{4/3}\eps^{2/3}} = \left(\frac{k^{3/4}}{\alpha \eps}\right)^{2/3} \cdot \left(\frac{\sqrt{k}}{\alpha^2}\right)^{1/3} \leq \frac{2}{3}\left(\frac{k^{3/4}}{\alpha \eps}\right) + \frac{1}{3}\left(\frac{\sqrt{k}}{\alpha^2}\right)$$
	by the AM-GM inequality, and our statement simplifies to $m = \Omega\left(\tfrac{k^{3/4}}{\alpha \eps} + \tfrac{\sqrt{k}}{\alpha^2}\right)$.
\end{proof}

\begin{lemma}
    Let $p$ be a distribution over $[k]$ such that $\tv{p}{U_k} = \alpha$ and let $G_1, \ldots, G_n$ be a uniformly random partition of $[k]$ into $n > 1$ subsets of size $\Theta(k/n)$. Define induced distribution $p_n$ over $[n]$ by $p_n(j) = \sum_{i \in G_j} p(i)$ for each $j \in [n]$. Then, with probability $\geq \tfrac{1}{954}$ over the selection of $G_1, \ldots, G_n$, $$\tv{p_n}{U_n} = \Omega\left(\alpha\sqrt{\tfrac{n}{k}}\right).$$
\end{lemma}
\begin{proof}
    It is equivalent to sample $G_1, \ldots, G_n$ as follows: randomly partition $[k]$ into $n/2$ same-size subsets $G_1', \ldots, G_{n/2}'$ (for neatness, we assume $n$ is even), and then randomly halve each of those to produce $G_1$ and $G_2$ (from $G_1'$), $G_3$ and $G_4$ (from $G_2'$), and so on.  We use the following lemma from~\citet{ACFT19} to connect the distances induced by $\{G_a'\}_{a=1}^{n/2}$ and $\{G_b\}_{b=1}^n$. Here, for a set $S$ we let $p(S)$ denote the total probability mass of set $S$, $p(S) = \sum_{s \in S} p(s)$.
    
    \begin{lemma}[Corollary 15 in~\citet{ACFT19}]
    \label{lem:half}
        Let $p$ be a distribution over $[k]$ with $\tv{p}{U_k} \geq \alpha$, and let $U$ be a random subset of $[k]$ of size $k/2$. Then $\P{U}{|p(U) - 1/2| \geq \tfrac{\alpha}{\sqrt{5k}}} > \tfrac{1}{477}$.
    \end{lemma}
    
    Slightly more generally, the proof of Lemma~\ref{lem:half} shows that for any distribution $p$ over $[k]$ and $S \subset [k]$, if $\tfrac{1}{2}\sum_{i \in S}|p(i) - \tfrac{1}{k}| \geq \alpha'$, and we choose a random subset $S' \subset S$ of size $\tfrac{|S|}{2}$, then $\P{S'}{|p(S') - \tfrac{p(S)}{2}| \geq \tfrac{\alpha'}{\sqrt{5|S|}}} > \tfrac{1}{477}$. 
    
    Fix the choice of $G_1', \ldots, G_{n/2}'$. For each $a \in [n/2]$, let $\alpha_a = \tfrac{1}{2}\sum_{i \in G_a'} |p(i) - \tfrac{1}{k}|$, the portion of $\tv{p}{U_k}$ contributed by $G_a'$. Replacing $\alpha'$ with $\alpha_a$ and $|S|$ with $k/(n/2)$ above, for each $a \in [n/2]$,
    $$\P{}{\left|p(G_{2a-1}) - \frac{p(G_a')}{2}\right| \geq \alpha_{a}\sqrt{\tfrac{n}{10k}}} \geq \frac{1}{477}.$$
    $p(G_{2a-1}) + p(G_{2a}) = p(G_a')$, so 
    $$\P{}{|p(G_{2a-1}) - p(G_{2a})| \geq 2\alpha_{a}\sqrt{\tfrac{n}{10k}}} \geq \frac{1}{477}.$$
    Then by triangle inequality
    $$\P{}{\left|p(G_{2a-1}) - \frac{1}{n}\right| + \left|p(G_{2a}) - \frac{1}{n}\right| \geq 2\alpha_{a}\sqrt{\tfrac{n}{10k}}} \geq \frac{1}{477}$$
    and in particular
    $$\E{}{\left|p(G_{2a-1}) - \frac{1}{n}\right| + \left|p(G_{2a}) - \frac{1}{n}\right|} \geq \frac{2\alpha_a}{477}\sqrt{\frac{n}{10k}}.$$
    
     For each $b \in [n]$ define $Y_b = \min\left(\left|p(G_b) - \tfrac{1}{n}\right|, \alpha_{\lceil b/2 \rceil}\sqrt{\tfrac{n}{10k}}\right)$. Let $Y = \sum_{b=1}^n Y_b$. First, we can lower bound $\E{}{Y}$, over the choice of $G_1', \ldots, G_{n/2}'$ and $G_1, \ldots, G_n$, as
    \begin{align*}
        \E{}{Y} =&\; \sum_{b=1}^n \E{}{\min\left(\left|p(G_b) - \frac{1}{n}\right|, \alpha_{\lceil b/2 \rceil}\sqrt{\frac{n}{10k}}\right)} \\
        \geq&\; \sum_{b=1}^n\tfrac{\alpha_{\lceil b/2 \rceil}}{477}\sqrt{\tfrac{n}{10k}} \\
        =&\; \tfrac{2\alpha}{477}\sqrt{\tfrac{n}{10k}}
        \stepcounter{equation}\tag{\theequation}\label{eq:first_fact}
    \end{align*}
    where the inequality uses the expectation lower bound above.
    
    Second, by definition of $Y_b$, $\max(Y) \leq \sum_{b=1}^n\alpha_{\lceil b/2 \rceil}\sqrt{\tfrac{n}{10k}} = 2\alpha\sqrt{\tfrac{n}{10k}}.$ Now assume for contradiction that $\P{}{Y \geq \tfrac{\alpha}{477}\sqrt{\tfrac{n}{10k}}} < \tfrac{1}{954}$. Then
    $$\E{}{Y} < \tfrac{\alpha}{477}\sqrt{\tfrac{n}{10k}} + \tfrac{\max(Y)}{954} \leq \tfrac{2\alpha}{477}\sqrt{\tfrac{n}{10k}}.$$
    Thus $\E{}{Y} < \tfrac{2\alpha}{477}\sqrt{\tfrac{n}{10k}}$, which contradicts Equation~\ref{eq:first_fact}. It follows that our assumption is false, and $\P{}{Y \geq \tfrac{\alpha}{477}\sqrt{\tfrac{n}{10k}}} \geq \tfrac{1}{954}$. The final claim follows from
\begin{align*}
	\frac{Y}{2} =&\; \frac{1}{2} \sum_{b=1}^n \min\left(\left|p(G_b) - \frac{1}{n}\right|, \alpha_{\lceil b/2 \rceil}\sqrt{\frac{n}{10k}}\right) \\
	\leq&\; \frac{1}{2}\sum_{b=1}^n \left|p(G_b) - \frac{1}{n}\right| \\
	=&\; \tv{p_n}{U_n}.
\end{align*}
\end{proof}

\section{Information Theory}
\label{sec:info}

\begin{definition}
Let $X$ be a random variable with probability mass function $p_X$. Then the \emph{entropy} of $X$, denoted by $H(X)$, is defined as $$H(X) = \sum_x p_X(x) \log\left(\frac{1}{p_X(x)}\right),$$
and the \emph{conditional entropy} of random variable $X$ conditioned on random variable $Y$ is defined as $H(X|Y) = \mathbb{E}_y[H(X|Y = y)]$. 
\end{definition}

Next, we can use entropy to define the mutual information between two random
variables.

\begin{definition}
\label{def:muinfo}
The \emph{mutual information} between two random variables $X$ and $Y$ is defined as $I(X;Y) = H(X) - H(X|Y) = H(Y) - H(Y|X)$, and the \emph{conditional mutual information} between $X$ and $Y$ given $Z$ is defined as $I(X;Y|Z) = H(X|Z) - H(X|YZ) = H(Y|Z) - H(Y|XZ)$. 
\end{definition}

\begin{definition}
The \emph{Kullback-Leibler divergence} between two random variables $X$ and $Y$ with probability mass functions $p_X$ and $p_Y$ is defined as
$$\kl{X}{Y} = \sum_x p_X(x) \log\left(\frac{p_X(x)}{p_Y(x)}\right).$$ 
\end{definition}

\begin{fact}
\label{fact:div}
Let $X,Y,Z$ be random variables, we have $$I(X;Y|Z) = \mathbb{E}_{x,z}[\kl{(Y| X = x, Z = z)}{(Y| Z = z)}].$$
\end{fact}

\begin{lemma}[Pinsker's inequality]
\label{lem:pinsker}
Let $X$ and $Y$ be random variables with probability mass functions $p_X$ and $p_Y$. Then
$$\sqrt{2\kl{X}{Y}} \geq 2\tv{p_X}{p_Y}.$$
\end{lemma}

\end{document}